\documentclass{article}
\usepackage{color}

\usepackage{amsmath,amssymb,amsfonts}
\usepackage{arxiv}
\usepackage[utf8]{inputenc} 
\usepackage[T1]{fontenc}    
\usepackage{hyperref}       
\usepackage{url}            
\usepackage{booktabs}       
\usepackage{amsfonts}       
\usepackage{nicefrac}       
\usepackage{microtype}      
\usepackage{cleveref}       
\usepackage{lipsum}         
\usepackage[sort&compress,numbers]{natbib}
\usepackage{doi}
\usepackage{graphicx,subfigure}

\usepackage{algorithmic}

\usepackage{textcomp}
\usepackage{multirow}
\usepackage{bm}
\usepackage{tabularx,makecell}
\usepackage{cuted}
\usepackage{multicol,caption}
\usepackage{siunitx}
\usepackage{dcolumn}

\usepackage{thmtools}

\newtheorem{lemma}{Lemma}

\newenvironment{proof}[1][Proof]{\begin{trivlist}
		\item[\hskip \labelsep {\bfseries #1}]}{\end{trivlist}}

\makeatletter
\newenvironment{equationate}{%
	\itemize
	\let\orig@item\item
	\def\item{\orig@item[]\refstepcounter{equation}\def\item{\hfill(\theequation)\orig@item[]\refstepcounter{equation}}}
}{%
	\hfill(\theequation)%
	\enditemize
}
\makeatother

\title{A survey of modularized backstepping control design approaches to nonlinear ODE systems}

\author{ \href{https://orcid.org/0000-0001-5522-8094}{\includegraphics[scale=0.06]{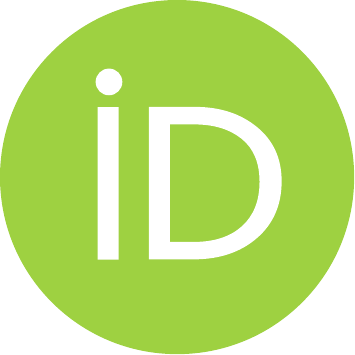}\hspace{1mm}Zhengru Ren}
	\\
	Institute for Ocean Engineering\\
	Shenzhen International Graduate School\\
	Tsinghua University\\
	Shenzhen, 518055, China \\
	\texttt{zhengru.ren@sz.tsinghua.edu.cn} \\
}

\hypersetup{
	pdftitle={A survey of modularized backstepping control design approaches to nonlinear ODE systems},
	pdfsubject={},
	pdfauthor={Zhengru Ren},
	pdfkeywords={},
}

\begin{document}
\maketitle

\begin{abstract} 
Backstepping is a mature and powerful Lyapunov-based design approach for a specific set of systems. Throughout the development over three decades, innovative theories and practices have extended backstepping to stabilization and tracking problems for nonlinear systems with growing complexity. The attractions of the backstepping-like approach are the recursive design processes and modularized design. A nonlinear system can be transferred into a group of simple problems and solved it by a sequential superposition of the corresponding approaches for each problem. To handle the complexities, backstepping designs always come up with adaptive control and robust control. The survey aims to review the milestone theoretical achievements among thousands of publications making the state-feedback backstepping designs of complex ODE systems to be systematic and modularized. Several selected elegant methods are reviewed, starting from the general designs, and then the finite-time control enhancing the convergence rate, the fuzzy logic system and neural network estimating the system unknowns, the Nussbaum function handling unknown control coefficients, barrier Lyapunov function solving state constraints, and the hyperbolic tangent function applying in robust designs. The associated assumptions and Lyapunov function candidates, inequalities, and the deduction key points are reviewed. The nonlinearity and complexities lay in state constraints, disturbance, input nonlinearities, time-delay effects, pure feedback systems, event-triggered systems, and stochastic systems. Instead of networked systems, the survey focuses on stand-alone systems. 
\end{abstract}
	

\section{Introduction}
First developed in 1991, backstepping approach~\cite{kanellakopoulos:1991systematic,kokotovic:1992joy} is a Lyapunov-based recursive design procedure for triangular strict-feedback form systems. In brief, the main idea is to divide the entire system into several subsystems and to recursively design a global controller according to the former subsystem. Adaptive control is involved when uncertainties are included in the system. The approach has been widely used in a wide scope of nonlinear systems. The controller and adaptive update laws can be designed simultaneously which improves the transient performance.

Over three decades of research and development, a significant amount of research papers extend the application field of backstepping intensively. In the first five years, the fundamental backstepping schemes and redesigns with parametric uncertainties attracted great attention. Later, the research focuses are the methodology improvement and more advanced parametric uncertainties. In the 2000s, the emphasis is laid on systems with uncertainty and nonsmoothness, neural adaptive control, time-delay effects, etc. Innovative design techniques provide solutions to various nonsmooth effects and complex uncertainties. In the recent decade, the scope of backstepping design methodology has been broadened to more complex systems, e.g., system with several aforementioned uncertainties, switched systems, stochastic systems, time-triggered systems, multi-agent systems, etc. 

Survey papers and textbooks on backstepping and adaptive control are reported in literature, e.g., \cite{krstic:1995nonlinear,kokotovic:2001constructive,khalil:1996noninear,ioannou:1996robust}. However, a considerable amount of studies have been conducted thereafter, such as, actuator and sensor nonlinearities \cite{tao:1996adaptive}, nonsmooth system and input nonlinearities \cite{tao:2013adaptive}, partial differential equation (PDE) \cite{krstic:2008boundary}, neural adaptive control \cite{ge:2013stable}, fault-tolerant control \cite{shen:2017fault}. There still lacks a systematically organized summary of state-of-the-art advances. Although thousands of papers are available online, a limited number of significant contributions can be extracted from them. Each method handles a specific form of nonlinear systems. The innovations normally rely on a series of assumptions. Since a lot of research combines several previous approaches, most pages for these works focus on some repeated assumptions, lemmas, definition, Lyapunov function candidates (LFC), and recursive deductions. In this survey, we try to point out the essence of different methods with designer-friendly characteristics. 

Therefore, we believe that it is timely and helpful to present a pervasive survey of state-of-the-art backstepping-like techniques to ordinary differential equation (ODE) systems with various nonlinearities. 
Another target of this survey is to help the interesting readers find suitable approaches to handle complex nonlinear system in a systematical and modularized way. 
A complex nonlinear system can be separated into several simplified problems and solved with the integration of basic methodologies and the corresponding Lyapunov function components.
The selection of specific approaches is determined by the constructive Lyapunov function. 
The design procedures are summarized as follows:  (i) divided it into several typical mini problems and transfer each of the mini problems into a typical form; (ii) select the reasonable approach to every mini problem, as well as the corresponding assumptions and Lyapunov function candidates; (iii) design the backstepping-like control recursively by integrating all approaches. 
Admittedly, the references provided in this survey do not give an exhausted list. Some conceptually deeper methods in the 1990s are not included since they are difficult to be categorized systematically, and the complexities in deduction and proof prevent their modularized applications. But it is believed to contain many important approaches that are welcomed by the successors in the academic domain.

The majority of the present survey was written in the final phase of my PhD study in 2019 and was slightly revised in 2020. It has been used as a lecture note in a PhD course at the Norwegian University of Science and Technology, Norway. Very positive feedback was received from the PhD candidates, who claimed that understanding advanced backstepping algorithms can never be so simple. 

The rest of this paper is organized as follows. 
In Section II, several widely-used elegant methods are reviewed, including the integrator backstepping, the dynamic surface control, command filtered backstepping, finite-time control, fuzzy logic system, neural network, Nussbaum function, barrier Lyapunov function, and hyperbolic tangent function.
In section III, key theoretical and methodological innovations to a class of systems are presented and discussed, i.e., state constraints, disturbance, input nonlinearities, and time-delay effects, pure-feedback, and stochastic systems. 
Most complexities are transferred to the forms which can be solved by the elegant methods. 
In Section IV, the advancements are summarized and discussed. 

To keep the review short and organized, some simplifications will be adopted hereafter without any specifications. Firstly, only single-input-single-output (SISO) systems are reviewed since the design methods for multiple-input-multiple-output (MIMO) systems are always similar. The results from MIMO systems can be adjusted accordingly. Furthermore, the state estimator design \cite{Marino1993output,FREEMAN1996735}, which is needed for output-feedback control, is neglected since it is not the emphasis of this survey. The estimated error should be involved additionally in the LFC to prove the stability of the entire system.

\subsection{Notations}
The Euclidean norm of a vector $x$ and Frobenius norm of a matrix $A$ are denoted by $|x|$ and $|A|_F$, respectively. Blackboard bold variables denote sets and domains, e.g., $\mathbb{R}^n$, $\mathbb{R}_+$, $\mathbb{R}^{n\times m}$ are the n-dimensional real space, nonnegative real number set, and $n\times m$ real space, respectively. $a^{(b)}$ stands for the $b^{th}$ derivative of $a$. $C^i$ denotes the set of all functions with continuous $i^{th}$ differential and and $C^{i,1}$ stands for the family of all nonnegative functions $V(x,t,r)$ on $\mathbb{R}^{n}\times\mathbb{R}_+\times S$ which are $C^i$ in $x$ and $C^1$ in $t$. The identify matrix and zero matrix with a size $m\times n$ are denoted by $I_{m\times n}$ and $0_{m\times n}$, respectively. Specially, $I_n$ and $0_n$ are the identity matrix and zero matrix short for $I_{n\times n}$ and $0_{n\times n}$, respectively. A set of indices is defined as $\mathcal{I}=\{1,2,\cdots,n-1\}$. Lie derivative of function $V$ in respect of $f$ is defined as $L_f V = \frac{\partial V}{\partial x}f(x)$. 

Variables with hat and tilde operators stand for the estimated values and estimate errors, respectively. Their relation is, e.g., given by $\tilde{\theta} = \theta-\hat{\theta}$. When $\theta$ is a constant, $\dot{\tilde{\theta}}=-\dot{\hat{\theta}}$. Overlines and underlines (e.g., $\overline{c}$ and $\underline{c}$), are the maxima and minima of a variable $c$. The sign function is given by $$
\text{sgn}(a) = \begin{cases}
	1, &\text{if } a>0 \\
	0, &\text{if } a=0 \\
	-1, &\text{if } a<0.
\end{cases}$$


\subsection{Preliminary} \label{section:paper13.preliminary}
Since the unknown nonlinearities are difficult to be fully estimated and canceled, the most widely-used stability lemma is updated as follows.
\begin{lemma}\label{lemma:paper13.boundedV} 
	A LFC $V(x)$ is	bounded if the initial condition $V(0)$ is bounded, $V(x)$ is positive definite and continuous and if a Lyapunov-like inequality holds, i.e.,
	\begin{equation}\label{eq:paper13.boundedVLemma.V}
		\dot{V}(x) \leq -\gamma V(x) + \delta,
	\end{equation}
	where $\gamma > 0$ and $\delta > 0$. Define $\rho:=\delta/\gamma$,
	\begin{equation}\label{eq:paper13.boundedVLemma.Vt_bound}
		0\leq V(t) \leq \rho +  (V(0)-\rho)\exp(-\gamma t).
	\end{equation}
	And it implies that
	\begin{equation}\label{eq:paper13.boundedVLemma.Vt}
		V(t)\leq e^{-\gamma t}V(0)+\int_{0}^{t}e^{-\gamma (t-\tau)}\rho(\tau)d\tau, \ \forall t\geq 0,
	\end{equation}
	for any finite constant $\gamma$.
\end{lemma}
\begin{proof}
	Times $\exp(-\gamma t)$ to both sides of (\ref{eq:paper13.boundedVLemma.V}), yields $\dot{V}\exp(-\gamma t) + \gamma V \exp(\gamma t)\leq \delta\exp(\gamma t)$. Integrating both side yields $\frac{d}{d\,t}V\exp(\gamma t) \leq \frac{\delta}{\gamma}\frac{d}{d\,t} \exp(\gamma t)$. Therefore, we have $V(t)\exp(\gamma t) - V(0) \leq \rho (\exp(\gamma t) - 1)$. 
\end{proof}

Lemma~\ref{lemma:paper13.boundedV} is significant to prove the global semi-stability which is the foundation of a great number of publications. Many works in this survey utilize Lemma~\ref{lemma:paper13.boundedV} to prove the final stability. The tracking errors converge to a small neighborhood of the origin, instead of exactly stays at the equilibrium point. The value of $\gamma$ can be found using a superficial inequality $c_{\min} \sum_{i=1}^{n}x_i^2 \leq c_1 x_1^2 + \cdots + c_n x_n^2 \leq c_{\max} \sum_{i=1}^{n}x_i^2$, where $c_{\min}=\min\{c_1,c_2,c_3,\cdots c_n\}$ and $c_{\max}=\max\{c_1,c_2,c_3,\cdots c_n\}$. From \eqref{eq:paper13.boundedVLemma.Vt}, the selection of $\gamma$ not only influences the final tracking error limitation, but also determines the convergence rate. 
There are several transformations are developed for more complex systems in similar forms.

Besides, there are many other method to prove the boundness; however, Lemma~\ref{lemma:paper13.boundedV} is the most widely-used with a low deduction complexity and relatively loose restrictions, and it inspires many other systematic methods.

\section{Elegant methodologies}
Backstepping approaches can be categorized into cancellation-based backstepping and passivity-based backstepping \cite{skjetne:2004robust}. According to the number of publications, the former method stands for the ruling position.
In this section, several elegant design methodologies from which thousands of publications benefit are reviewed; see Fig.~\ref{fig:paper13.elegant_approaches}. These methods perform as basics to a significant amount of publications. 
Each of them could be applied to various complexities and nonlinearities. 
We consider a method is elegant if it can be applied to various problems and cooperate with other design approaches.
The development of each method starts from an integrator chain and is extended to increasingly complex systems, e.g., from linear systems to nonlinear systems, from fully known systems to unknown and disturbed systems, and from second- or third-order systems to recursive design for large-scale systems. Since most nonlinearities cannot be fully canceled, most of the methods are prevalent to Lemma~\ref{lemma:paper13.boundedV}, and the design procedures and stability proof follow a similar methodology.

There are two dominant directions to do the cancellation. The first one is to cancel the nonlinearity as much as possible, resulting in a robustness approach. The second way is to approximate everything using online neural networks (NNs) or fuzzy logic systems (FLSs) and cancel the estimates.

\begin{figure*}
	\includegraphics[width=0.7\linewidth]{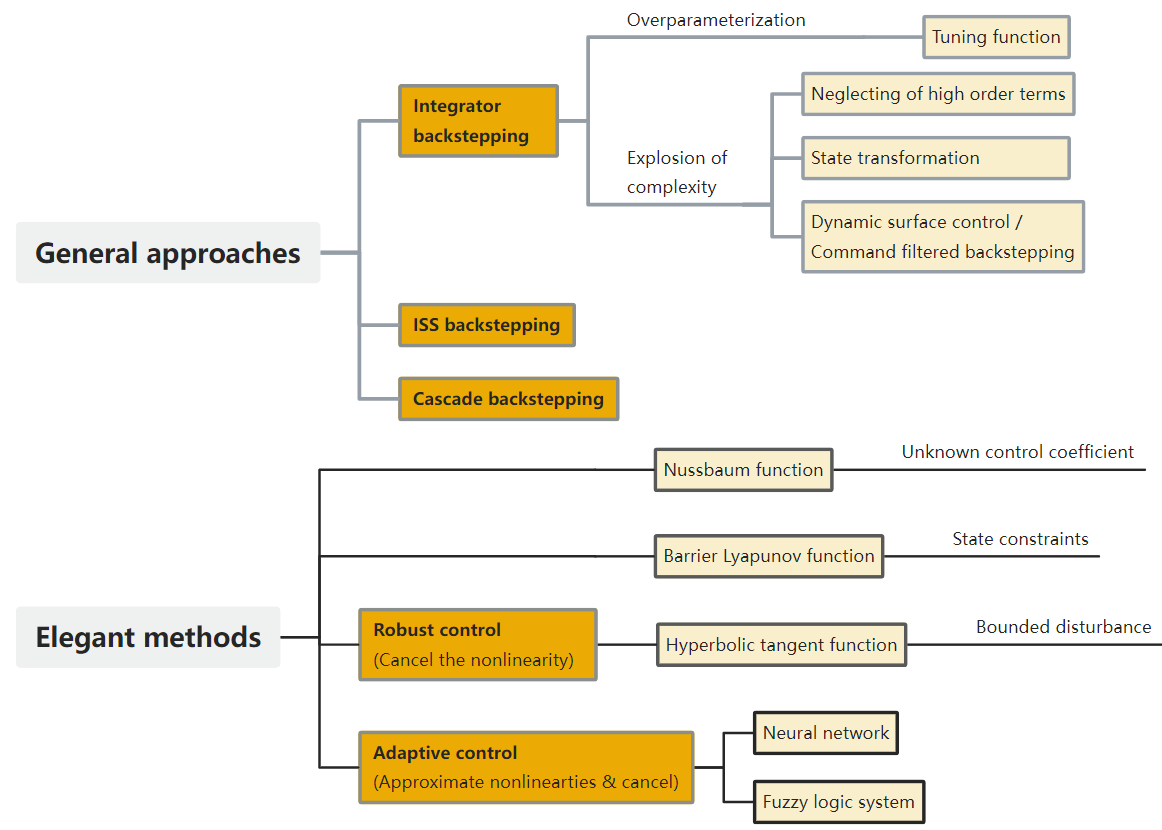}
	\caption{Elegant modularized methods used in backstepping designs.}
	\label{fig:paper13.elegant_approaches}
\end{figure*}

\subsection{General design approaches}
Being the starting point of backstepping designs, the strict-feedback nonlinear system with a triangular structure is given by
\begin{subequations}\label{eq:paper13.strictFeedbackForm}
	\begin{align}
		\dot{x}_i &= f_i(\bar{x}_i) + g_i(\bar{x}_i)x_{i+1} + \phi_i(\bar{x}_i)^\top\theta, i\in\mathcal{I} \\
		\dot{x}_n &= f_n(\bar{x}_n) + g_n(\bar{x}_n) u +  \phi_n(\bar{x}_n)^\top\theta, \\
		y &= x_1,
	\end{align}
\end{subequations}
where 
$x_1, \cdots, x_n$ are the states,
$f_1,\cdots, f_n$, $g_1, \cdots, g_n$ are smooth functions, $g_i$ is the control coefficient function,
$\bar{x}_i=[x_1,x_2,\cdots,x_i]^\top$,
$\theta\in\mathbb{R}^{p}$ is an unknown constant vector, 
$\phi_i(\cdot):\mathbb{R}^i\to\mathbb{R}^{p}$ is assumed to be known, and
$u$ and $y$ are the input and output respectively. The terms $\theta^\top\phi_i(\bar{x}_i)$ are the parametric uncertainties, namely, the line-in-parameters condition \cite{Marino1993global}.
Without specific declarations, the output signal is $y = x_1$ hereafter. When $f_1 = \cdots = f_n=0$ and $g_1=\cdots=g_n=1$, the system is simplified to be an integrator chain, or, namely, the Brunovsky form. 
The error states is defined as $z_1 = x_1 - x_{1d}$ and $z_{i+1} = x_{i+1} - \alpha_i$ where $x_{id}$ is the desired trajectory and $\alpha_i$ is the virtual control law to be defined later. The control objective is to track $y$ to $x_{1d}$ as $t\to\infty$, i.e., $z_1 = x_1 - x_{1d}\to 0$ as $t\to\infty$. 
New variables are defined as
$\bar{x}_{1d}=[x_{1d}^{(1)},x_{1d}^{(2)},\cdots,x_{1d}^{(i)}]^\top$,
$\bar{z}_i=[z_1,z_2,\cdots,z_i]^\top$, and
$\bar{z}_{i:j}=[z_i,z_{i+1},\cdots,z_j]^\top$.

A common assumption for the backstepping tracking control is the sufficiently smooth desired trajectory, i.e., the reference signal $x_d(t)$ and its derivatives up to the required number of order are known, bounded, and continuous. In addition, the signs of $g_1,\cdots,g_n$ are assumed to be known and constant, i.e., $|g_i|>0$ for all $t$. Three fundamental methods will be reviewed in this section, including (adaptive) integrator backstepping, adaptive backstepping with tuning functions, and dynamic surface control/command filtered backstepping. Asymptotically tracking adaptive backstepping design for a system \eqref{eq:paper13.strictFeedbackForm} could be achieved. 

Research on each method starts from the integrator chain, where the virtual control coefficients are constantly $g_i = 1$. Later, the algorithms are enriched with control coefficient $g_i$. A common assumption is that the control coefficients are known with unchanged known signs, i.e., there exist unknown positive constants $\underline{g}_i$ and $\overline{g}_i$ s.t., $0<\underline{g}_i \leq |g_i|\leq \overline{g}_i$. Hence, the partial derivatives of $g_i$ are strictly either positive or negative. Singularity problem is also neglected.

\subsubsection{Integrator backstepping and ISS backstepping}
A commonly used quadratic Lyapunov function and its time derivative is given by $V_{i,QF} = \frac{1}{2}z_i^2$, and $\dot{V}_{i,QF} = z_i \dot{z}_i$. The Lyapunov function in step $i$ is the superposition of that in the former step and $V_{i,QF}$, i.e., $V_i = V_{i-1}+V_{i,QF}$.
The global asymptotic stability of the origin after substituting the control laws is verified in \cite{kanellakopoulos:1991systematic}. When an internal system is included in the system, $L_g V$ backstepping shares similar design procedures based on passivity theory \cite{arcak:2000robustification}. 

When the unknown parts can not be completely canceled, ISS stability is employed to prove the boundness of the tracking error\cite{freeman1998robustness}. The gains can be assigned to ensure the tracking accuracy \cite{jiang1994Small}. 

The two main drawbacks of integrator backstepping are: 

(i) Overparameterization problem: Since separate adaptive laws for $\dot{\hat{\theta}}_i$ are designed with every virtual control $\alpha_i$, an overparameterization problem is aroused. This problem is released by using tuning functions to estimate all the unknown parameters in the final step, which significantly reduces the controller's order \cite{krstic:1992adaptive}. However, the robustness of this method to nonparametric nonlinearities is weak.

(ii) Explosion of complexity: It results from the repeated differentiation of virtual control laws. The derivatives of the virtual control signals $\alpha_i$ are required. As the order of the nonlinear system increases, the calculation of the derivation becomes prohibitive. The analytic derivation is cumbersome and tedious when the system order $n$ is larger than 3. To avoid the explosion of complexity, there are a few methods.
First of all, neglecting of high order terms is a possible way to reduce the computational costs in practical implements; however, the Lyapunov stability can not be guaranteed theoretically and it even results in instability. Second, state transformation can be adopted by defining a series of states as $\tilde{x}_1 = x_1 - x_{1d}$ and $\tilde{x}_i = \dot{\tilde{x}}_{i-1}$ with the control objective $\tilde{x}_1\to 0$ as $t\to \infty$. Then, the system $\text{col}(\tilde{x}_1,\cdots,\tilde{x}_n)$ is in a integrator-chain form. The error states are defined $z_1 = \tilde{x}_1$ and $z_i = \tilde{x}_i - \alpha_{i-1}$ \cite{gayaka:2012global}. However, the complex time differentiation of $f_i$ is still inevitable. Thirdly, efforts, such as dynamic surface control (DSC) \cite{swaroop:1997dynamic} and command filtered backstepping \cite{farrell:2009command,yu:2015observer} have been made to avoid such problem systematically. 

\subsubsection{Cascade backstepping}
Instead of integrator backstepping, cascade backstepping is an alternative \cite{sorensen2020comparing}. For the dynamics $\dot{z}_i$, $z_{i+1}$ is considered as an input. The virtual control is designed based on Lyapunov function candidate $V_i = V_{i,QF}$.  If $z_i=0$ and system $(\dot{z}_1,\cdots,\dot{z}_{i-1})$ is stable, and if $z_{i+1}=0$ and the stability of system $\dot{z}_i$ is proved to be stable, then the stability of system $(\dot{z}_1,\cdots,\dot{z}_i)$ is proved as a cascade system.

\subsubsection{Dynamic surface control and command filtered backstepping}
At the early stage, dynamic surface control is an extension of multiple surface sliding control overcoming the explosion of terms satisfying semi-global boundedness of the tracking error. The studies focus on a less general system. Later, command filtered backstepping is proposed in a more general triangular form. The control law is similar to that uses tuning functions. The difference between these two approaches remains in the stability proofs. 
The error states also have similar definitions, i.e., $\tilde{x}_1 = x_1 - x_d$ and $\tilde{x}_i = x_i - x_i^c$, where $x_i^c$ is a filtered signal of $\alpha_i$. First-order lowpass filters are utilized to synthetic input at each step of the design process. The total number of the lowpass filter is $n-1$. The derivatives of the virtual control signal are replaced by the filtered signals; hence, analytic differentiation and the constant-$g_i$ assumption are released. Moreover, non-Lipschitz systems can be stabilized. The DSC methods can global exponentially stabilize the system and arbitrarily bounded tracking for Lipschitz systems, and semi-globally arbitrarily bounded regulation and tracking for non-Lipschitz system \cite{swaroop:2000dynamic}.

The commanded filters have various structures, for instance,
\begin{itemize}
	\item First-order lowpass \cite{farrell:2009command}
	\begin{equation}
		\dot{x}_i^c = - \omega_i(x_i^c-\alpha_{i-1}),
	\end{equation}
	\item Second-order lowpass filter \cite{yu:2015observer}
	\begin{subequations}
		\begin{align}
			\dot{\varphi}_{i,1} & = -\omega_n \varphi_{i,2},\\
			\dot{\varphi}_{i,2} & = -2\zeta\omega_n\varphi_{i,2}-\omega_n(\varphi_{i,1}-\alpha_i),
		\end{align}
	\end{subequations}
	\item First-order Levant differentiator with finite-time convergent property \cite{yu:2018finite}
	\begin{subequations}
		\begin{align}
			\dot{\varphi}_{i,1} & = -a_i |\varphi_{i,1}-\alpha_i|^{1/2} \text{sign}(\varphi_{i,1}-\alpha_i) + \varphi_{i,2},\\
			\dot{\varphi}_{i,2} & = -b_i \text{sgn}(\varphi_{i,2}- \dot{\varphi}_{i,1}),\ i\in \mathcal{I},
		\end{align}
	\end{subequations}
\end{itemize}
where $a_i$ and $b_i$ are tuned coefficients,
$\zeta$ is the damping ratio, 
$\omega_n$ is the natural frequency, and 
$x_{i+1}^c = \varphi_{i,1}$.

The virtual control laws and adaptive laws $\alpha_i$ are simplified with $\dot{x}_i^c$. The stability is proved with the LFC 
$V = \sum_{i=1}^{n}\frac{1}{2} v_i^2 +\frac{1}{2}\tilde{\theta}^\top\Gamma^{-1}\tilde{\theta}$ with $v_i = z_i -  \xi_i$ and $\xi_i$ is commonly defined as 
\begin{subequations}
	\begin{align}
		\dot{\xi}_1 & = -k_1\xi_1 + g_1 \xi_2 + (x_2^c-\alpha_2)g_1,\\
		\dot{\xi}_i & = -k_i\xi_i + g_i \xi_{i+1} - g_{i-1}\xi_{i-1} + (x_{i+1}^c-\alpha_i)g_i,\\
		\dot{\xi}_n & = -k_n\xi_n + g_{n-1} \xi_{n-1}.
	\end{align}
\end{subequations}
DSC and command filter backstepping can be used combining with other techniques. Instead of feedback cancellation, adding one power integrator technique is another possible solution \cite{CORON1991adding,lin:2000adding}.

\subsubsection{Parameter separation}
Young's inequality is the most widely used parameter separation technique in backstepping design. 
If there exists a term $z_i a_i$ in $\dot{V}_i$ where $a_i$ is a bounded variable or a constant, then we have $z_i a_i\leq \frac{1}{2}\varepsilon_i^2 z_i^2 + \frac{1}{2}\frac{a_i^2}{\varepsilon_i^2}$ with $\varepsilon_i>0$ according to Young's inequality. The term $\frac{1}{2}\varepsilon_i^2 z_i^2$ can be easily canceled by the virtual control law $\alpha_i$, and $\frac{1}{2}\frac{a_i^2}{\varepsilon_i^2}$ is left to the final Lyapunov function $V_{n}$ as $\delta = \sum_{i=1}^n \delta_i= \sum_{i=1}^n \frac{1}{2}\frac{a_i^2}{\varepsilon_i^2} \geq 0$. Hence, the stability can be proved when the parameters in $\delta$ are bounded. 
According to Lemma~\ref{lemma:paper13.boundedV}, when choosing sufficient large gains $c_i>0$, $\gamma$ is increased to reduce the boundary of $V_{n}$ and $\rho$ is reduced to ensure a smaller region for the errors. The selection of $\varepsilon_i$ does not influence the overall control gain $\bar{c}_i = c_i + \frac{\varepsilon_i^2}{2}$ and the final tracking performance. Smaller $\varepsilon_i$ results in lower $\frac{\varepsilon_i^2}{2}$ and larger $\delta_i$. To ensure the same boundary, higher $c_i$ is required to achieve higher $\gamma$ in \eqref{eq:paper13.boundedVLemma.V}, and vice versa. 

\subsection{Finite-time control}
It is known that the asymptotic stability implies the system trajectory converges to the equilibrium as $t\to\infty$ resulting in a slow convergence rate near the equilibrium. Sometimes, fast response and high-prevision tracking performance are preferred rather than asymptotic stability. Comparing with the typical backstepping approaches which grantees asymptotic convergence, the finite-time control technique has a faster response and better disturbance-rejection ability, which ensure the convergence to the equilibrium in a finite time, i.e., $\lim\limits_{t\to T} x(t,x_0) = 0$ and $x(t,x_0)=0$ for any $t\geq T$ where $T$ is the settle time. In addition, the finite-time stability has a higher disturbance-rejection capacity and allows a system converges to the origin even more than one equilibrium exist. At the early stage, local finite-time design approaches homogeneous approximation are limited to two- or three-dimensional systems \cite{hong:2001output}, e.g., terminal sliding mode (TSM) control \cite{yu:2005continuous}. Then, \cite{huang:2005global} extends the results to the global domain.

The origin is said to be a finite-time-stable equilibrium of an autonomous system $\dot{x} = f(x(t))$, $f(0)=0$, $x\in\mathbb{R}^{n}$ if the finite-time convergence and Lyapunov stability hold. A Lyapunov-based finite-time stability theorem is proposed in \cite{bhat:2000finite}. The origin is a finite-time-stable equilibrium if there exists a continuous positive definite function $V(x)$, real numbers $\gamma > 0$, and $r\in(0, 1)$, such that,
\begin{equation}\label{eq:paper13.finiteTime.V}
	\dot{V}(x) \leq - \gamma V^r(x), \forall x\in \mathbb{N}\backslash\{0\},
\end{equation}
where $\gamma>0$ and $\mathbb{N}$ is an open neighborhood of the origin \cite{haddad:2008finite}. 
The settling-time function is a function of the initial value of the LF $T \leq \frac{1}{\gamma(1-r)}V(x_0)^{1-r}, x\in\mathbb{N}$. 
The selection of $r$ determines the convergence rate of the Lyapunov function. Compared to the scenario $r=1$ (Lemma~\ref{lemma:paper13.boundedV}), the convergence rate for $r\in(0, 1)$ is faster in the close neighborhood of the origin, while it is slower in the region away from the origin. 
In addition, the extended forms of \eqref{eq:paper13.finiteTime.V} are 
\begin{equationate}
	\item $\bullet$ $\dot{V}(x) \leq -\gamma_1 V(x) - \gamma_2 V^r(x)\label{eq:paper13.finiteTime.V2}$ \cite{yu:2005continuous}
	\item $\bullet$ $\dot{V}(x) \leq -\gamma_1 V^{r'}(x) - \gamma_2 V^r(x)\label{eq:paper13.finiteTime.V3}$ \cite{sun2017new}
\end{equationate}
where $\gamma_1>0$, $\gamma_2>0$ and $r'>1$. They ensure higher convergence rates when the tracking error are far away from the origin, and the settle times are $T = \frac{1}{\gamma_1(1-r)}\ln \frac{\gamma_1 V^{1-r}(x_0) + \gamma_2}{\gamma_2}$ for \eqref{eq:paper13.finiteTime.V2} and $T=\frac{1}{\gamma_2(1-r)}+\frac{V^{1-r'}(x_0)-1}{\gamma_1(1-r')}$ for \eqref{eq:paper13.finiteTime.V3}. For a stochastic system, the finite-time stability in probability guaranteed if $\mathcal{L} V(x) \leq - \gamma V^r(x)$ holds \cite{khoo:2013finite,wang:2015finite}.

The finite-time stability could be extended by involving a finite constant $\delta>0$ in $\dot{V}$. Similar to (\ref{eq:paper13.boundedVLemma.V}),
Practical finite-time stability is guaranteed if $\dot{V}(x) \leq -\gamma_1 V(x) - \gamma_2 V^r(x) +  \delta$ holds \cite{yu:2018finite}. 
However, the definition ``practical'' in finite-time stability could be problematic when considering all systems with asymptotical stability could reach a certain value in a finite time \cite{ren2020finite}. Asymptotically stability only fails to converge to zero in a finite time, and the convergence speed could be tuned by assigning the control gains.

Similar to the typical backstepping methods, the finite-time backstepping approach requires the terms to be canceled by the fractional-order terms. Hence, one of the following assumptions is always required.
\begin{itemize}
	\item $|f_i| \leq (\sum\limits_{j=1}^{i}|x_j|)\rho_i(\bar{x}_i)$ \cite{huang:2005global}
\item $|f_i| \leq \frac{1}{2}|x_{i+1}|^{r_i} + \sum\limits_{j=1}^{i}|x_j|\rho_i(\bar{x}_i)) \sigma$ \cite{hong:2006adaptive}
\item $|f_i|\leq \sigma \rho_i(\bar{x}_i))$ \cite{hong:2006finite}
\item $|f_i| = \varphi (t) \sum\limits_{j=1}^{i}|x_j|^{m_{ij}} + \sigma\sum\limits_{j=1}^{i}|x_j|^{n_{ij}} $ (time-varying system) \cite{zhang:2012finite}
\end{itemize}
where $\rho_i(\bar{x}_i)$ are smooth known $C^1$  positive functions and $\sigma\geq 1$ is an uncertain constant. The forth assumption is used to handle a time-varying system.
The orders of all systems should be decided before the design process. The designed parameters are always $r_1 > \cdots> r_n$ since the higher-dimension dynamics should react faster than the lower dimension.

The selection of LFCs for finite-time stability could be 
\begin{itemize}
\item $\text{[FT1]}$ $V_{i,FT} = \int_{\alpha_i}^{x_i}\left( s^{1/q_k} - \alpha_i^{1/q_i} \right) d s$
\item $\text{[FT2]}$ $V_{i,FT} = \int_{\alpha_i}^{x_i}\left( s^{\beta_{i-1}} - \alpha_{i-1}^{\beta_{i-1}} \right) d s$ \cite{hong:2006finite}
\item $\text{[FT3]}$ $V_{i,FT} = \int_{\alpha_i}^{x_i}\left( s^{1/q_k} - \alpha_i^{1/q_i} \right)^{2-q_k}d s$ 
\cite{huang:2005global}
\item $\text{[FT4]}$ $V_{i,FT} = \int_{\alpha_i}^{x_i}\left( s^{1/q_k} - \alpha_i^{1/q_i} \right)^{2-q_k-\tau}d s$ (order-1) where $\tau$ is a ratio of two numbers.
\end{itemize}

Besides the recursive design method, the finite-time stabilization is always designed with an inductive approach, i.e., the virtual control laws $\alpha_i$ are designed accordingly with assumed known form $\dot{V}_{i-1}$ and let $\dot{V}_{i}$ and $\dot{V}_{i-1}$ have a similar inductive form \cite{hong:2006adaptive}. This is because of the complex forms of the control law and yielded time derivative of the Lyapunov function.
Most designs are according to the Young's-inequality-like inequalities in \cite{qian:2001continuous}.
Another application of the finite-time stability and recursive design is the finite time disturbance observer \cite{du:2013recursive}.

\subsection{Estimation-based backstepping - neural network and fuzzy logic system}\label{sec:NN_FLS}
It is always difficult to design functions $\phi_i$ in \eqref{eq:paper13.strictFeedbackForm}. The estimation-based approach is a solution to this problem. The most widely used approximation theories using together with backstepping are the NNs and FLSs. Their orientations and theories are different. However, their deduction in backstepping, mathematically, are very similar. Hence, both methods are presented together in this section.

NNs are widely used in various backstepping designs to approximate the uncertainties, namely, neural adaptive control. According to the universal approximation property, any smooth function in a compact set can be approximated by an NN with an arbitrarily small error by a sufficiently large number of nodes. In other words, the arguments of an unknown function starting from any initial compact superset $\Omega^0$ remains within a compact set $\Omega$. Briefly, the idea is to estimate all nonlinearities by NNs and cancel them by the virtual control laws.

Two NNs are widely used in literature, they are:
\begin{itemize}
\item Two-layer radial boundary function NN (RBFNN) \cite{ge2002NN}
\begin{equation}
	f(z) = W^\top S (z) + \varepsilon,
\end{equation}
\item Multilayer neural networks (MNN) \cite{zhang:1999design,zhang2000adaptive} and three-layer wavelet NN (WNN) \cite{ho:2005adaptive}
\begin{equation}
	f(z) = W^\top S(D^\top z + T) + \varepsilon = W^\top S(U^\top \bar{z}) + \varepsilon,
\end{equation}
\end{itemize}
where
$W=[w_1,\cdots,w_l]^\top \in\mathbb{R}^{l}$ and $V=[v_1,\cdots,v_l]^\top \in\mathbb{R}^{q \times l}$ are the first-to-second and second-to-third layer weight vectors,
$S(z)=[s_1(z),\cdots,s_l(z)]^\top$, 
$T=[t_1,\cdots,t_l]^\top \in\mathbb{R}^{l}$, and 
$\varepsilon$ is the corresponding reconstruction error.
MNN and WNN share a similar format by using state transformation $U^\top = [D^\top,T]$ and
$\bar{z} = [z^\top, 1]^\top$. The difference among various NNs is the selection of basis function $s_i$. For a RBFNN, $s_i(z)$ is defined as
$s_i(z) = \exp \left[-\frac{(z-\mu_i)^\top(z-\mu_i)}{r_i^2} \right]$,
where $\mu_i=[\mu_{i1},\cdots,\mu_{in}]^\top$ is the center of the receptive field and $r_i$ is the width of the Gaussian function.
Then, function $f$ is approximated as $\hat{f} = \hat{W}^\top S (Z)$ and $\hat{W}^{\top} S(\hat{U}^{\top}\bar{z})$. The fundamental assumption for the neural adaptive control is the boundness of the weight matrix $W$.

The ideal weights $W$ are assumed to be bounded with known $W_m$, i.e., $\|W\|_F \leq W_m$. The NN approximation is compensated by the virtual control. Define the error vector of weights as $\tilde{W}=W - \hat{W}$, a quadratic term $\frac{1}{2}\tilde{W}^\top \tilde{W}$ is added to the LFC. To ensure the value of $\gamma$ in Lemma~\ref{lemma:paper13.boundedV} is sufficiently large, a $\sigma$-modification is adopted, i.e., an extra $- \Gamma \sigma \hat{W}$ with $\sigma>0$ is added in the adaptive update law $\dot{\hat{W}}$ which is similar to $\dot{\hat{\theta}}$ in integrator backstepping. The crucial inequalities in the final stability proof are $\tilde{W}^\top\hat{W} = \tilde{W}^\top W - \tilde{W}^\top\tilde{W}\leq \frac{1}{2}(-\tilde{W}^\top\tilde{W} + W^\top W)$ and $\text{tr}\{\tilde{W}^\top\hat{W}\} \leq \frac{1}{2}(-|\tilde{W}|_F + |W|_F)$, which are proved by $\tilde{W}^\top W \leq \frac{1}{2}(\tilde{W}^\top\tilde{W} + {W}^\top W)$. Combining with the assumption of bounded $W$, the terms $-\frac{1}{2}\sigma \tilde{W}^\top\tilde{W}$ and $\frac{1}{2}\sigma W^\top W$ are merged into $-\gamma V$ and $\delta$ in Lemma~\ref{lemma:paper13.boundedV}, respectively.
The relation also holds to $\theta$, that is, $\tilde{\theta}^\top\hat{\theta}\leq \frac{1}{2}(-\tilde{\theta}^\top\tilde{\theta} + \theta^\top \theta)$. Again, they are merged to $-\gamma V$ and $\delta$ in Lemma~\ref{lemma:paper13.boundedV}, respectively.

FLS is a ratio of a linear combination of the input variables membership functions and the sum of the membership functions, given by
$f(x) = \theta^\top \varphi(x) + \varepsilon$ or $f(x) = W^\top S(x) + \varepsilon$ and $\varepsilon \leq \bar{\varepsilon}$ where $\bar{\varepsilon}$ is a bounded constant,
$\theta = [\bar{y}_1,\cdots,\bar{y}_N]^\top$,
$\bar{y}_l = s_i = \max\limits_{y\in \mathbb{R}} \mu_{G^l}(y)$,  
the fuzzy-membership function $\mu_{G^l}(x_i)=\exp[-\frac{x_i-a_i^l}{b_i^l}]$,
$\varphi^\top = [\varphi_1,\cdots,\varphi_N]$, and
the fuzzy basic function $\varphi_l = \frac{\prod_{i=1}^{N} \mu_i^l(x_i)}{\sum_{l=1}^{N} \prod_{i=1}^{N} \mu_i^l(x_i)}$
\cite{lewis:1999deadzone}.
Function $\varphi$ can be the Fourier series expansion which has a similar form to three-layer NNs \cite{chen:2010adaptivePeriodicDistrb}. The important properties are $0 \leq \varphi^\top \varphi \leq 0$ and $z_i \theta^\top \varphi \leq \frac{z_i^2}{4\lambda}|\theta|^2 + \lambda$.

The drawbacks of the estimation-based methods are the long learning time resulting from the significant number of NN nodes and adaption parameters in order to receive sufficient approximation accuracy. An explosion of states occurs to high-order systems resulting in difficult implementations. Most examples in the case studies are second-order systems. To reduce the size of the adaptive control question, $|W|$ can be used instead of $W$ \cite{yang:2004combined}. Additionally, the design procedure is based on the assumed existed upper bounds of the weight matrices. The stability is restricted to be local as the NN approximation is only valid in the specific compact sets. Another difficulty is that the term $\frac{\partial \alpha_i}{\partial \hat{\theta}_i}\dot{\hat{\theta}}_i$ for MIMO systems is a function of all states. It cannot be merged into other terms which to be estimated by the following NNs.

Another shortage of using NNs and FLSs is the lack of capacity to extract the underlying structures of the nonlinear functions, though the estimation fits well with the nonlinear function. The typical adaptive control can be applied to identify the system by using some preset important governing terms. If the basis function library is well guessed, the function structure can be identified explicitly. On the other hand, the estimate of an NN or FLS does not require a well-guesses basis function library. 

\subsection{Nussbaum function}\label{sec.paper13.NussbaumFunction}
In the aforementioned sections, a tricky assumption for adaptive backstepping design is the unchanged sign of the known control gain, namely, constant control direction. Nussbaum function was firstly proposed in \cite{nussbaum:1983some} and introduced to adaptive backstepping design by \cite{ye:1998adaptive}. It has been used to design control laws for systems with unknown $g_i$, for example, with a lack of priori knowledge of control direction.

A function $\mathcal{N}(\chi)$ is said to be a Nussbaum-type function if
$\lim\limits_{s\to\infty} \sup \frac{1}{s} \int_{0}^{s} \mathcal{N}(\chi) d\chi = + \infty$ and $\lim\limits_{s\to\infty} \inf \frac{1}{s} \int_{0}^{s} \mathcal{N}(\chi) d\chi = - \infty$ hold. Its integral is defined as $\mathcal{M}(\chi):=\int_{0}^{\chi} \mathcal{N}(\tau) d \tau$. .

There are two types of Nussbaum-type functions, i.e., the amplitude-elongation and the time-elongation\cite{chen:2016saturated}. Amplitude-elongation Nussbaum-type functions are commonly adopted which are the products of a class $\mathcal{K}_\infty$ function and a trigonometric function, for example,
\begin{itemize}
\item $\mathcal{N}(\chi) = \chi^2 \cos(\chi)$,
\item $\mathcal{N}(\chi) = \chi^2 \sin(\chi)$,
\item $\mathcal{N}(\chi) = \exp(\chi^2) \cos(\chi \pi /2)$ \cite{nussbaum:1983some},
\item $\mathcal{N}(\chi) = \frac{1}{2}e^{-\sigma \chi} \cos\chi - \frac{1}{2}e^{\sigma\chi}\cos\chi$ \cite{chen:2015adaptive},
\item $\mathcal{N}(\chi) = \cosh(\lambda\xi) \sin(\xi)$ \cite{chen:2014adaptive},
\item $\mathcal{N}(\chi) = \exp(\xi^2/2)(\xi^2+2)\sin(\chi)$ \cite{ding:2015adaptive}.
\end{itemize}

The virtual control law is designed to change the time derivative of LFC to a form, for example, in $V(t) \leq c_0 + e^{-c_1 t}\int_{0}^{t} (g_i(t) \mathcal{N}(\chi) + 1)\dot{\chi} e^{c_1\tau}d\tau$. By multiplying $e^{c_1t}$ to both sides, the boundedness of term $g_1 z_1 z_2$ can be proved with the Lemmas proposed in \cite{ye:1998adaptive,ge:2003robust}. Similar to Lemma~\ref{lemma:paper13.boundedV}, an extension show the boundness on $[0,t_f)$ of $V(t)$, $\chi_i(t)$, and $\sum_{i=1}^{n}\eta_i\int_{0}^{t}\mathcal{N}(\chi_i(\tau))\dot{\chi}_i(\tau) d\,\tau$ is given \cite{chen:2014adaptive}, if the following inequality holds
\begin{equation}
\begin{aligned}
	V(t)\leq &\sum_{i=1}^{n}\eta_i\int_{0}^{t}\mathcal{N}(\chi_i(\tau))\dot{\chi}_i(\tau) d\,\tau \\ &+ \sum_{i=1}^{n}a_i\int_{0}^{t}\dot{\chi}_i(\tau)\,\tau + \delta ,\ \forall t\in[0,t_f),
\end{aligned}
\end{equation}
where $a_i$ and $\eta_i$ are constants with $a_i>0$, $\eta_i$ has the same sign, and $|\eta_i|\in[\eta_{\min},\eta_{\max}]$. 


The design is achieved by replacing $g_i \alpha_i$ by $g_i \alpha_i= g_i \mathcal{N}(\chi) \bar{\alpha}_i = (g_i \mathcal{N}(\chi)-1) \bar{\alpha}_i + \bar{\alpha}_i$ and canceling the nonlinearities by designing $\alpha_n$. Then the update law of $\chi$ is assigned to be $\dot{\chi} = z_i \bar{\alpha}_i$ and stability is proved. The variable $\chi$ does not enter the LFC.

The Nussbaum gain technique has been successfully applied in many problems presented in section~\ref{sec:paper13.single_system}.
A drawback of Nussbaum gain is that the disturbances can not be compensated directly. 


\subsection{Barrier Lyapunov function}\label{sec:BLF}
Inspired by \cite{ngo:2005integrator} and proposed in \cite{tee:2009barrier}, barrier Lyapunov function (BLF) is very popular LFC tool in backstepping-like design. It is compatible with nearly all techniques in this survey. Significant amounts of works use BLF to handle the asymptotic stabilization problem with partial or full state constraints. Unlike a typical Lyapunov function holding a globally radically unbounded property, a BLF goes to infinity when any state approaches its barriers.  

A BLF is a scalar function $V_{BLF}(x):\mathbb{D}_x\to\mathbb{R}_+$, defined with respect to the system $\dot{x}=f(x)$ on an open region $\mathbb{D}_x$ containing the origin, that is (i)
continuous and positive definite, (ii) has continuous first-order partial derivatives at every point of $\mathbb{D}_x$, 
(iii) has the property $\dot{V}(x) \to \infty$ as $x$ approaches the boundary of $\mathbb{D}_x$, and (iv) satisfies $V_{BLF}(x) \leq b$ for all $t \geq 0$ along the solution of $\dot{x}=f(x)$ for $x(0)\in \mathbb{D}_x$ and some positive constant $b$.
The BLFs are presented in Figure~\ref{fig:paper13.SymmetricAsymmetricBarrierLyapunovFunction}.
For any positive constants $k_{ai}$ and $k_{bi}$, let $\mathbb{D}_{zi}:=\{z_i \in \mathbb{R}| -k_{ai} < z_i < k_{bi} \} \subset \mathbb{R}$, the BLF $V_{i,BLF}(z_i) \to \infty$, as $z_i \to -k_{ai}$ or $z_i \to k_{bi}$. 
The tracking error $z_i(t)$ is proved to be remained in the open set $z_i \in \mathbb{D}_{zi}$ for all $t\geq 0$ if $\dot{V} \leq 0$ holds. Again, $z_i(t)$ remains bounded for all $t\geq 0$ if the inequality $\dot{V} \leq -\gamma V + \delta$ holds \cite{tee:2009barrier}.

The are several ways to construct BLFs depending on the type of the constraints. Three widely-used BLFs and their time derivatives are listed as follows:
\begin{itemize}
\item $\text{[BLF1]}$ Symmetric barriers ($k_{ai}=k_{bi}$) 
\begin{itemize}
	\item $V_{i,BLF}(z_i) = \frac{1}{2} \log \frac{k_{bi}^2}{k_{bi}^2-z_i^2} = \frac{1}{2} \log \frac{1}{1-\xi_b^2}$,
	\item $\dot{V}_{i,BLF}(z_i) = \frac{z_i}{k_{bi}^2-z_i^2}\dot{z}_i$, 
\end{itemize}
\item $\text{[BLF2]}$ Asymmetric barriers, ($k_{ai}\neq k_{bi}$)
\begin{itemize}
	\item $V_{i,BLF}(z_i)=\frac{q}{p}\log \frac{k_{bi}^p}{k_{bi}^p-z_i^p}+\frac{1-q}{p}\log \frac{k_{a}^p}{k_{ai}^p-z_i^p}$, 
	\item $\dot{V}_{i,BLF}(z_i)= \frac{q z_i^{p-1}}{k_{bi}^p-z_i^p}\dot{z}_i+ \frac{(1-q) z_i^{p-1}}{k_{ai}^p-z_i^p}\dot{z}_i$,
\end{itemize}	
\item $\text{[BLF3]}$ Time-varying constrain situation ($\dot{k}_{ai}\neq 0$ and $\dot{k}_{bi}\neq 0$)
\begin{itemize}
	\item $V_{i,BLF}(z_i)= \frac{q}{p}\log \frac{1}{1-\xi_b^p}+\frac{1-q}{p}\log \frac{1}{1-\xi_a^p}$, 
	\item $\dot{V}_{i,BLF}(z_i)= \frac{q \xi_b^{2p-1}}{k_{bi}(1-\xi_b^2p)}(\dot{z}_i-\frac{z_i}{k_{bi}}\dot{k}_{bi})+\frac{(1-q) \xi_a^{2p-1}}{k_{ai}(1-\xi_a^2p)}(\dot{z}_i-\frac{z_i}{k_{ai}}\dot{k}_{ai})$,
\end{itemize}
\end{itemize}
where $\log(\cdot)$ is the natural logarithm,
$p\geq n$ is an even integer,
switching function $q(z_i)=\begin{cases}1, &if\ z_i>0\\	0, &if\ z_i\leq 0	\end{cases}$,
$\xi_a = \frac{z_i}{k_{ai}}$, and 
$\xi_b = \frac{z_i}{k_{bi}}$. 
Then, BLF3 can be rewritten into a simpler form $
V_{i,BLF}(z_i)= \frac{1}{2p}\log\frac{1}{1-\xi^{2p}}$, where $\xi = q \xi_b + (1-q)\xi_a$. Novel BLF functions are given by $V_{i,BLF} = \frac{k_{bi}^2}{\pi}\tan(\frac{\pi z_i^2}{2k_{bi}^2})$ \cite{xu:2013state} and $V_{i,BLF} = \cot\frac{\pi}{2}(1-(\frac{z_i}{k_{bi}})^2)$ \cite{chen:2017adaptive}.

\begin{figure}
\includegraphics[width=0.8\linewidth]{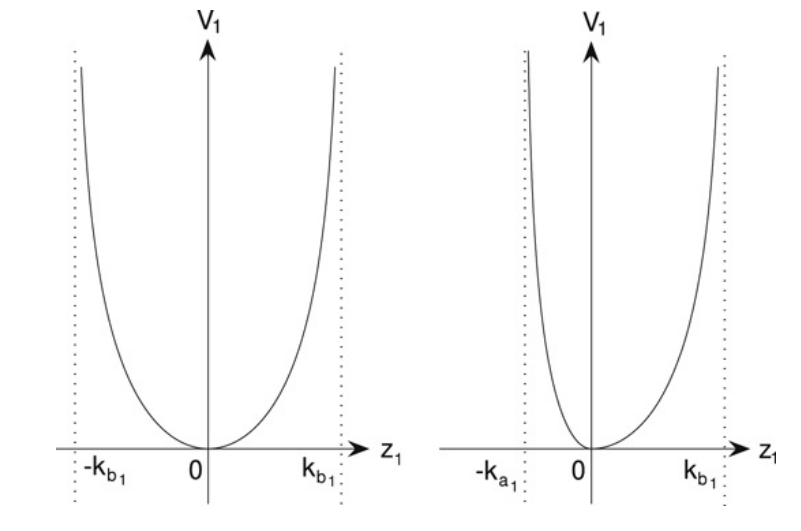}
\caption{Barrier Lyapunov functions.}
\label{fig:paper13.SymmetricAsymmetricBarrierLyapunovFunction}
\end{figure}

The BLF-based method can be applied based on two assumptions. 
First, the control coefficients $g_i$ is constantly positive or negative without changing sign when $x_1(t) \in \mathbb{D}_{x_1}$, i.e., $|g_i(\bar{x_i})|>g_0>0$ where $g_0$ is a positive constant.
Second, there exist positive constants $A_0$, $\underline{x}_{1d}$, $\overline{x}_{1d}$, and $\overline{x}_{id}$, $i=2,\cdots,n$ satisfying $\max\{\underline{x}_{1d},\overline{x}_{1d}\} \leq A_0\leq k_{c1}$ and $- \underline{x}_{1d} \leq x_{1d}(t)\leq \overline{x}_{1d}$ and $|x_{1d}^{(i)}(t)| < \overline{x}_{id}$ for all $k_{c1}>0$ and $t\geq 0$ \cite{tee:2009barrier,ren:2010adaptive}. 
In the design procedures, a BLF, for example in step $i$, is used when the error state $z_i$ is constrained i.e., $V_i=V_{i-1} + V_{i,BLF}$. Otherwise, quadratic LFC is adopted. 

For a system which is not perfectly known, an important inequality, which holds for all $|\xi|<1$ and any positive integer $p$ \cite{ren:2010adaptive}, is given by
\begin{equation}\label{eq:paper13.BLF.ineq1}
\log\frac{1}{1-\xi^{2p}} < \frac{\xi^{2p}}{1-\xi^{2p}}.
\end{equation}
For any positive constant $k_{bi}$, substituting $\xi = \frac{z_i}{k_{bi}}$ yields $\log \frac{k_{bi}^2}{k_{bi}^2-z_i^2} < \frac{z_i^2}{k_{bi}^2-z_i^2}$ and $
\log \frac{k_{bi}^{2p}}{k_{bi}^{2p}-z_i^{2p}} \leq \frac{z_i^{2p}}{k_{bi}^{2p}-z_i^{2p}}$.
Eq.~(\ref{eq:paper13.BLF.ineq1}) simplifies the control law by removing the logarithm operator and avoid the over-overparameterization problem \cite{liu:2016barrier}.
When system unmodeled dynamics exists, applying the inequality removes $k_{b1}^2-z_1^2$ in the virtual control and simplified the control law.
The numerator $-\frac{c_i z_i^2}{k_{bi}^{2p}-z_i^{2p}}$ in $V_n$ is replaced by $-c_i\log\frac{ z_i^2}{k_{bi}^2-z_i^2} = -c_i V_{i,BLF}$. Hence, the boundness can be proved according to Lemma~\ref{lemma:paper13.boundedV} \cite{ren:2010adaptive}.  

For an MIMO system, any positive constant vector $b_i\in\mathbb{R}^{n_i}$,
the following inequality holds for any vector $z_i\in\mathbb{R}^{n_i}$ in the
interval $|z_i| < |k_{bi}|$, $\log \frac{k_{bi}^\top k_{bi}}{k_{bi}^\top k_{bi} -z_i^\top z_i} \leq \frac{z_i^\top z_i}{k_{bi}^\top k_{bi} - z_i^\top z_i}$ \cite{zhao:2014adaptive}.
Then, the error signals are bounded, $|z_2|,\cdots,|z_n| \leq \sqrt{2V(0)\exp(-\rho t)}$, and
$k_{a1}(t) (1-e^{-2pV(0)\exp(-\rho t)})^{\frac{1}{2p}} \leq z_1  \leq k_{b1}(t) (1-e^{-2pV(0)\exp(-\rho t)})^{\frac{1}{2p}}$.
As $t \to \infty$, $|z_1(t)| \leq k_{b1} \sqrt{1-\exp(-2(\rho +  (V(0)-\rho)\exp(-\gamma t)))}  \leq k_{b1} \sqrt{1-\exp(-2\rho)}$. Hence, the upper limit of the tracking error $z_1$ decreases with $\rho$. Therefore, $|z_1|$ can be arbitrary small when $\gamma$ is sufficiently large.


Large control action may result when the states approach the boundary of the boundaries. Another drawback of BLF is the upper and lower limits are assumed to be known which are probably unknown in practical applications. Another concern is the conservative initial state. The initial states have to stay in the constraints, i.e., $x_i(0) \in (-k_{ai},k_{bi})$. At the region near the barrier, $V_i$ will be extremely large. Integral barrier Lyapunov functional (iBLF) $V_{i,iBLF} = \int_{0}^{z_i} \frac{\sigma k_{ai}^2}{k_{ai}^2-(\sigma+\alpha_{i-1})^2} d\sigma$ significantly relaxes the feasibility conditions \cite{tee:2012control}. The properties of iBLF are 
(i) $\frac{z_i^2}{2} \leq V_{i,iBLF} \leq z_i^2 \int_{0}^{1} \frac{\beta k_{ai}^2}{k_{ai}^2-(\beta z_i + \text{sgn}(z_i)A_{i-1})^2} d\beta$ and  
(ii) $\dot{V}_{i,iBLF}  
= \frac{k_{ai}^2 z_i}{k_{ai}^2-x_i^2} \dot{z}_i + z_i(\frac{k_{ai}^2 }{k_{ai}^2-x_i^2} - \rho_i)\dot{\alpha}_{i-1}$, 
where 
$A_i<k_{ai}$ is the upper bound of $|x_{id}|$ and 
$\rho_i = \frac{k_{ai}}{2 z_i}\ln \frac{(k_{ai}+z_1+x_{id})(k_{ai}-x_{id})}{(k_{ai}-z_i-x_{id})(k_{a1}+x_{id})})\dot{x}_{id}$.

\subsection{Hyperbolic tangent function}\label{section:paper14.thanh}
Backstepping-like cancellation is convenient to cancel the terms of states, e.g., $z_i^2$. However, it fails to cancel the unknown scalar terms, e.g., the external disturbance $d_i$. 
Compared with a complex projection operator used in early studies, an innovative inequality is adopted to handle the parametric bounded disturbance. If the disturbance is bounded, the backstepping design process main relays on an inequality for any $\varepsilon > 0$ and $\eta\in\mathbb{R}$, i.e.,
\begin{equation}
0 \leq |\eta| - \eta \tanh(\frac{\eta}{\varepsilon}) \leq k_p \varepsilon,
\end{equation}
where $k_p$ is a constant that satisfies $k_p = e^{-(kp+1)}$, e.g., $k_p = 0.2758$. 
The aim of the inequality is to transfer the product of a scalar and the error state, i.e., $d_i z_i $, into a form with a term of explicit $z_i$ and bounded constant in $\dot{V}_i$. The first part can be easily canceled and the later part can be integrated into $\delta$ in \eqref{eq:paper13.boundedVLemma.V}.
The error between the absolute value and its hyperbolic tangent approximation is presented in Figure~\ref{fig:paper13.operator.renzr_tanh}. The introduced $\tanh(\cdot)$ is compensated by the virtual control law in the recursive design, and the final LFC can be proved to be bounded using Lemma~\ref{lemma:paper13.boundedV}. 

The system uncertainties and model errors can be equivalent to a disturbance $d_i$, for example, the input-deadzone effects. In addition, a $\tanh(\cdot)$ function is used to approximate the symmetric input saturation and absolute operator to prevent the nonsmoothness in the time derivatives.

\begin{figure}
\begin{center}
	\includegraphics[width=1 \linewidth]{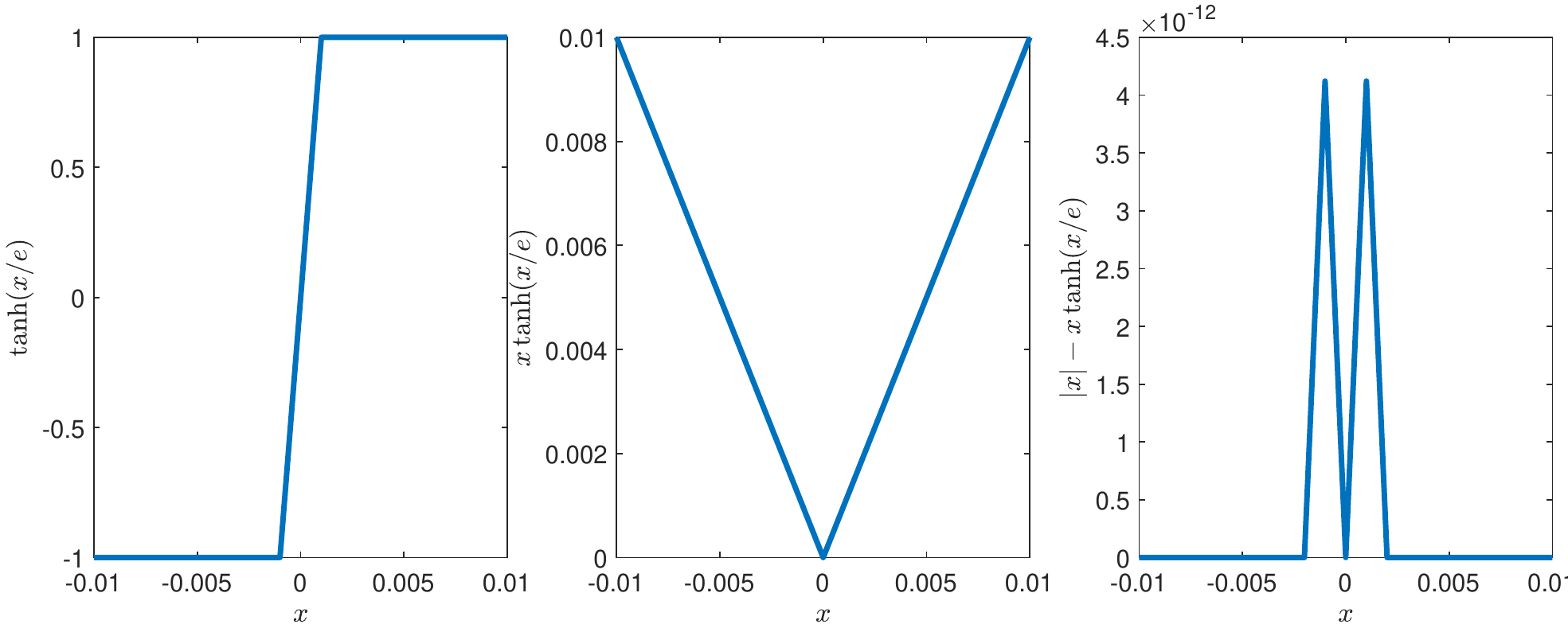}
	\caption{tanh($\cdot$) function ($e=0.001$).}
	\label{fig:paper13.operator.renzr_tanh}
\end{center}
\end{figure}

The singularity problem challenges the backstepping design. 
If a positive unknown term $a$ exists in $\dot{V}_i$ resulting in $-\frac{a}{z_i}$ in the virtual control, the system faces singularity problem when $z_i = 0$, resulting in the term $-\frac{a}{z_i}$ is impossible to be estimated. To prevent such issue, a property is given, i.e., for any constant $\eta>0$ and variable $z_i\in\mathbb{R}$, $\lim_{z_i \to 0} \frac{1}{z_i} \tanh^2(\frac{z_i}{\eta}) = 0$ \cite{ge:2007approximation}. The term $a$ can be modified to be $a = \frac{2}{z_i}\tanh^2(z_i/\eta) a z_i + \left(1-2\tanh^2(z_i/\eta)\right) a$. The former part $\frac{2}{z_i}\tanh^2(z_i/\eta) a z_i$ can be estimated. An addition useful property is $1-2\tanh^2(z_i/\eta)\leq 0$ when $|z_i|\geq 0.8814\eta$. Besides, $z_i$ is already bounded by $\eta$ if $|z_i|< 0.8814\eta$.

Another property of $\tanh(\cdot)$ function is that inequality $1-16\tanh(\frac{z_i}{v_i}) \leq 0$ holds for $|z_i| \leq 0.2554 v_i$ \cite{ge:2007approximation}. It is used to compensate the radical unbounded term in $\dot{V}_i$, e.g., $z_i^2 = [1-16\tanh(\frac{z_i}{v_i})]z_i^2 + 16\tanh(\frac{z_i}{v}_i) z_i^2$ and the latter term $16\tanh(\frac{z_i}{v_i}) z_i^2$ can be approximated by the NNs. 

Furthermore, $\tanh(\cdot)$ function prevents singularity which will be discussed later.

\section{Systematic backstepping designs to various nonlinearities}\label{sec:paper13.single_system}
In this section, backstepping designs applying to a class of systems are discussed. Systems with various nonlinearities are transferred to specific forms that are suitable for the aforementioned elegant approaches. A summary is given in Fig.~\ref{fig:paper13.nonlinearities}.

\begin{figure*}
	\includegraphics[width=0.9\linewidth]{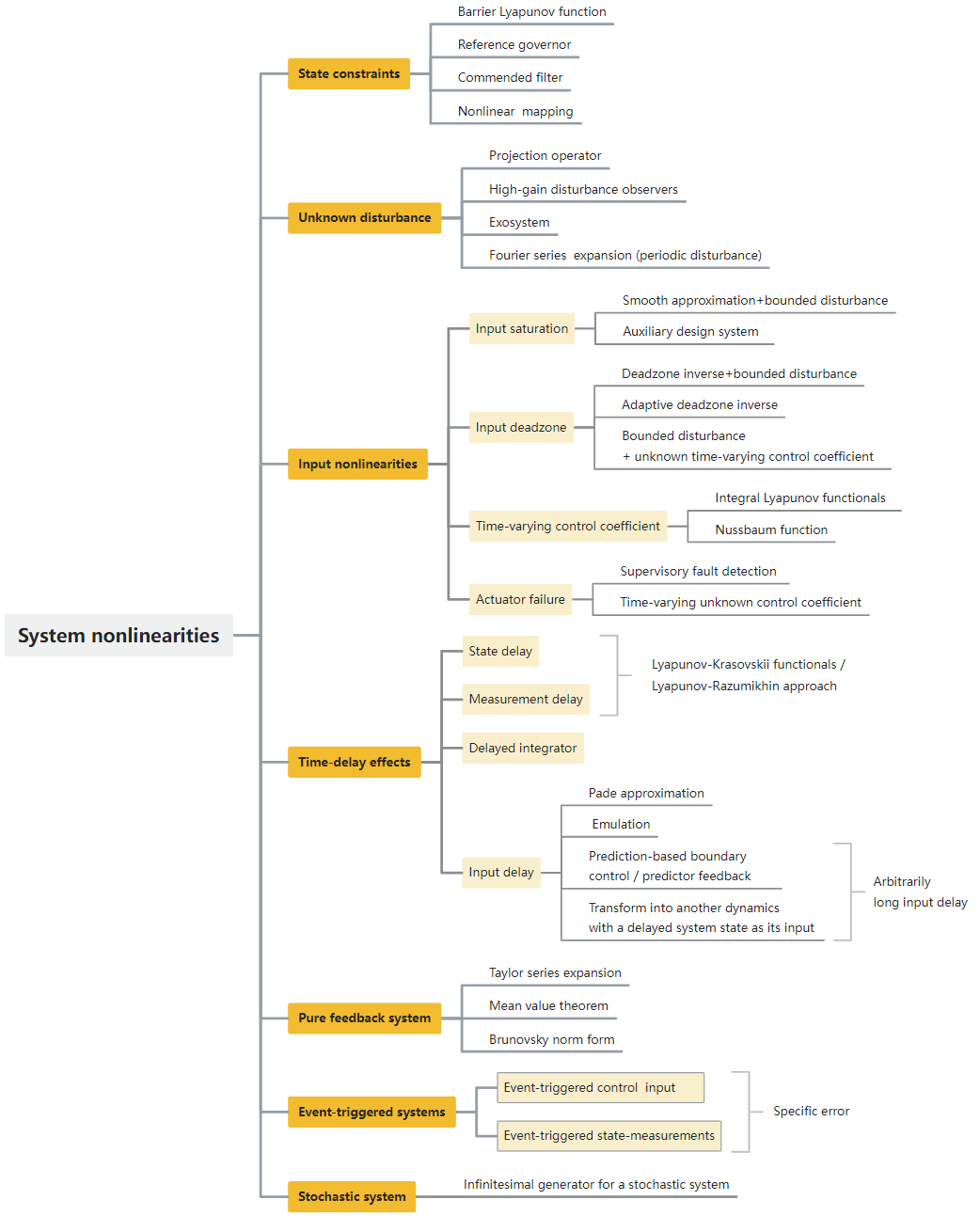}
	\caption{A summary of nonlinearities and corresponding design method.}
	\label{fig:paper13.nonlinearities}
\end{figure*}

\subsection{State constraints}
Due to the existence of actuator physical limitation and operational limits, there is a clear boundary that the states or the tracking error can reach in some applications, i.e., $x(t) \in \mathcal{D}_x$ or $z_1(t)=y(t)-y_d(t)\in \mathcal{D}_{z1}$. The state constraints are due to physical limitations or performance requirements. For example, in system (\ref{eq:paper13.strictFeedbackForm}), the output $y(t)$ is required to remain in the set $|y| \leq k_{c1}$ for all $t \geq 0$, where $k_{c1}$ is a positive constant. However, typical backstepping design methods cannot ensure state constraints. Besides model predictive control, there are a few Lyapunov methods to handle state constraints.

BLF is the most popular approach due to its systematic deduction and simple form. Hundreds of paper adopt BLF to handle state constraints. This has been intensively reviewed in Section~\ref{sec:BLF}.

The reference governor is another method to handle the state constraint problem. To avoid violating the system constraints, the reference signal is modulated using online optimization
\cite{gilbert:2002nonlinear}. However, this method is not specific for adaptive backstepping design. 

The commended filter can be used by regulating the changing rate of $\alpha_i$ the states. This is achieved by modifying the filters with several saturation operators \cite{sonneveldt:2007nonlinear}. However, state constraints are not proved. The saturation operators in the filter may deteriorate the stability to the entire system through critical damping is selected.

State constraints can be canceled by introducing a nonlinear mapping, which transfers a constrained nonlinear system into a pure-feedback form without outputs or state constraints \cite{zhang:2017adaptiveNeural}.

\subsection{Unknown disturbance}\label{sec:disturbance}
Disturbance is a very common unmodeled dynamics in a practical system due to the external environment. Besides, the estimate error due to modeling uncertainty is always modeled as unknown disturbance. A typical disturbed nonlinear system is given by
\begin{subequations}\label{eq:paper13.backstepping.boundedDisturbance.chen2011_adaptive.system}
	\begin{align}
		\dot{x}_i &= f_i(\bar{x}_i) + g_i(\bar{x}_i)x_{i+1} + d_i, \ i=\mathcal{I},\\
		\dot{x}_n &= f_n(\bar{x}_n) + g_n(\bar{x}_n) u+ d_n,
	\end{align}
\end{subequations}
where $d_i$ ($i=1,\cdots,n-1$) and $d_n$ are the mismatched and matched disturbances, respectively. Disturbances are impossible to be directly measured, resulting in the robustness violation and failure of the direct compensation approaches. A disturbance is unmatched when the disturbance $d_i$ enters the system with a different state from the control input $u$; otherwise it is matched. A system with unmatched disturbance can be converted into one with matched disturbance using equivalent input disturbance technique \cite{ding:2008asymptotic}. 

The designs commonly depend on bounded-disturbance assumptions. 
The widely-used assumptions are summarized as follows:
\begin{itemize}
	\item $|d_i|\leq \bar{d}_i$ \cite{koshkouei:2000adaptive},
\item $|d_i|\leq \rho_i(\bar{x}_i)\theta_i$ \cite{polycarpou:1996robust},
\item $|d_i|\leq (\rho_{i1}(\bar{x}_i) + \rho_{i2}(\bar{z}_i))\theta_i$ \cite{tong:2010adaptive},
\end{itemize}
The first one assumes a upper limits $\overline{d}_i$ for $d_i$. Using the second and third assumptions,  the disturbance can be separated into unknown bounded constants $\theta_i$ and known smooth functions $\rho_i(\bar{x}_i) \in\mathbb{R}_+$ for $t > t_0$. Disturbance rejection control can be categorized into robustness-based and estimation-based (or approximation-based) methods.

\subsubsection{Robustness-based methods}
Projection operator is widely used at the early stage of the robust adaptive control. A projection operator for two vectors $\theta,\ y \in \mathbb{R}^k$ is defined as
\begin{equation}
\text{Proj}(\theta,y,f)=\begin{cases}
y-\frac{\bigtriangledown f(\theta) \bigtriangledown f(\theta)^\top}{\|\bigtriangledown f(\theta)\|^2} y f(\theta),\quad  \\ \ \ \ \  \text{if}\ f(\theta)>0 \bigwedge y^\top \bigtriangledown f(\theta)>0, \\
y, \qquad \text{otherwise},
\end{cases}
\end{equation}
where $f:\mathbb{R}^k\to \mathbb{R}$ is a convex function and  $\bigtriangledown f(\theta_b)=[\frac{\partial f(\theta)}{\theta_1} \cdots \frac{\partial f(\theta)}{\theta_k}]^\top$.
The core property of the projection operator $\tilde{\theta} (\text{Proj}(\theta,y,f) - y) \leq 0$ can be used to design adaptive update law. Besides the basic projection operator, there are several different projection operators, such as $\Gamma-$Projection \cite{Lavretsky:2011projection} and parameter projection \cite{ikhouane:1998adaptive,Grip:2015globally}. The exact structure of the convex function $f$ is of no importance. When using projection operators in the backstepping design, sufficient differentiability is required. The projection operators cooperate with other elegant methods. However, it is no longer popular anymore after the 2000s. A reason is the complex calculation of the terms with Laplace operator $\bigtriangledown f$. 

Combining with Lemma~\ref{lemma:paper13.boundedV}, simpler control algorithms are proposed, e.g., \cite{koshkouei:2000adaptive}. A more popular method is the $\tanh(\cdot)$ function presented in Section~\ref{section:paper14.thanh}. When the bounded uncertain term $d_i$ appear, a major part of the term $|d_i||z_i| $ in $\dot{V}_i$ can be compensated and the uncancellable parts are merged into $\delta$ in \eqref{eq:paper13.boundedVLemma.V}.

\subsubsection{Estimation-based methods}
Though the system disturbance cannot be directly measured, several techniques are available to estimate the disturbance with measurements. The overall stability is proved with a combination of the control law and observer. Hence, the estimate error should be included in the Lyapunov function.

As mentioned before, NNs and FLSs can be used to approximate all sort of bounded nonlinearities. Hence, they have been intensively investigated and widely applied in literature. We skip here since they have been reviewed them in Section~\ref{sec:NN_FLS}.

Several high-gain disturbance observers are designed based on linear system theories with similar forms
\cite{chen:2010sliding,won:2015high}. In another way, a lowpass filter is used as an disturbance observer by filtering the error $\dot{z}_{i}-(f_i+g_i-\dot{\alpha}_{i-1})$. However, $\dot{z}_i$ measurements are not always available. A possible solution is to replace $\dot{z}_i$ with $-\mu_i z_i$, and the estimated disturbance  $\hat{d}_i$ is the output from the lowpass filter plus $\mu_i z_i$ \cite{rashad:2016novel}. The LFC for Step $i$ is modified to be $V_i = V_{i-1} + V_{i,QF} + \frac{1}{2} \tilde{d}_i^2$, and $\hat{d}_i$ is canceled by the virtual control law directly. 

According to the internal model principle, deterministic disturbances with a priori known waveform can be expressed in a form an exosystem \cite{li:2006internal}, the disturbance can be compensated by invoking the exosystem in the feedback loop. The problem is firstly solved by proposing an nonlinear observer of the unknown external disturbance and using a state-feedback adaptive regulator \cite{nikiforov:1998adaptive}. Different from NN/FLS, parameter-dependent observer for unknown disturbance based on exosystem has a simpler structure. The bounded disturbance can be presented as $d = \vartheta + \delta_{d}$, where $\delta_{d}$ is a bounded function, $\vartheta$ is the output of a linear exosystem given by
\begin{subequations}\label{eq:paper13.exosystem}
\begin{align}
\dot{\chi} &= \Gamma \chi, \\
\vartheta &= c^\top \chi,
\end{align}
\end{subequations}
where 
$\chi\in\mathbb{R}^{q}$ is the state vector, 
the constant $q \times q$ matrix $\Gamma$ has all its eigenvalues on the imaginary axis, and $c$ is a constant vector. Without loss of generality, the pair $(c^\top,\Gamma)$ is assumed to be observable. Here, the dimension of the exosystem $q$ is a design parameter, $c$ and $\Gamma$ are unknown. Similar to the estimation-based approach, higher-order exosystems are adopted when complex disturbances are involved. The application of this disturbance observer technique can be  found in both unmatched and matched disturbances \cite{nikiforov:2001Nonlinear,sun:2014composite}.

A specific case is the periodic disturbance. Periodic disturbances often exist in many mechanical control systems such as industrial robots and marine systems. Fourier series expansion can be adopted to model the disturbance, i.e., $d(t)=B^\top \Phi(t)+\delta_d(t)$ with bounded truncation error $|\delta_d(t)|\leq \overline{\delta}_d$, where $\Phi=[\phi_1,\cdots,\phi_q]^\top$, $\phi_1=1$, $\phi_{2j}=\sqrt{2}\sin(2\pi jt/T)$ and $\phi_1=1$, $\phi_{2j+1}=\sqrt{2}\cos(2\pi jt/T)$ \cite{chen:2010adaptivePeriodicDistrb}.

However, a common challenge for the observer-based approaches is the process to model a rapidly changing disturbance. The observer may react not fast enough to estimate the high-frequency disturbances. It is argued that if the observer dynamics are much faster than the system dynamics, the disturbance is very likely to be estimated and canceled by the controller. The training time in addition depends on system complexity. The stability of the entire closed-loop system can be deteriorated due to the large model mismatch.

\subsection{Input nonlinearities}
For a system with perfect actuator performance, the actuator output $u$ equals to the control signal $v$, i.e., $u(t)=v(t)$. However, $u$ is a function of $v$ due to the actuator characteristics and limitations in the form of $u(v(t-\tau_d(t)))$ where $\tau_d$ is the time delay. In this section, several input nonlinearities are reviewed, assuming $\tau_d=0$.

\subsubsection{Input saturation}
Physical actuators surely have their limits, also called input constraints. When the control input remains within the boundness, the input saturation effects are negligible. The main challenges for input saturation is the nonsmoothness at the breakpoints and limited scope of input. To stabilize a system with bounded control input has significant meaning to practical industrial applications. A general form of input saturation is given by
\begin{equation}
u(v) = \text{sat}(v) = \begin{cases}
u_{min} , & \text{if }  v\leq u_{min}, \\
g_s(v), & \text{if }  u_{min} < v < u_{max}, \\
u_{max} , & \text{if }  v\geq u_{max},
\end{cases}
\end{equation}
where 
$g_s$ is a smooth function, and
$u_{min}$ and $u_{max}$ are the minimum and maximum values for $u$.
For a simplest linear saturation with symmetric limits, i.e., $g_s(v) = v$ and $-u_{min}= u_{max}=u_m>0$, the input saturation is simplified to be
\begin{equation}\label{eq:paper13.linearSaturation}
u(v(t)) = \text{sat}(v(t)) = \begin{cases}
v(t), & \text{if }  |v|<u_m, \\
\text{sgn}(v(t)) u_m, & \text{if }  |v|\geq u_m.
\end{cases}
\end{equation}

Smooth approximations are adopted to overcome the nonsmoothness. 
\begin{itemize}
\item The symmetric saturation (\ref{eq:paper13.linearSaturation}) is approximated by \cite{zhou:2006adaptiveSaturation,zhou:2017adaptive}
\begin{equation}
\text{sat}(v)\approxeq \eta_{sat}(v) = u_m\tanh(\frac{v}{u_m})=u_m \frac{e^{v/u_m}-e^{-v/u_m}}{e^{v/u_m}+e^{-v/u_m}}.
\end{equation}

\item The approximation for the asymmetric saturation is given by 
\begin{equation}
\text{sat}(v)\approxeq\eta_{sat}(v)=\frac{2 \overline{u}}{\pi} \arctan \left(\frac{\pi v}{2\overline{v}} \right),
\end{equation}
\end{itemize}
where $\overline{u} = u_{max}$, $\overline{v}=v_{max}$ if $u_{max}/v_{max}\geq u_{min}/v_{min}$ and $\overline{u} = u_{min}$, $\overline{v}=v_{min}$ if $u_{max}/v_{max}\leq u_{min}/v_{min}$ \cite{chen:2017adaptive}. 

The approximation difference $d_{sat}(v) = \text{sat}(v)-\eta_{sat}(v)$ is a bounded function with a upper limit $|d_{sat}(v)|\leq u_m (1-\tanh(1))=\overline{d}_{sat}$. 
There are a few approaches to continue after the smooth approximation, i.e., 
(i) Consider $d_{sat}$ as a bounded disturbance, then the robustness of the system can be proved with Lemma~\ref{lemma:paper13.boundedV} using disturbance rejection control shown in Section~\ref{sec:disturbance}.
(ii) The system is augmented with an extra state equation $\dot{v}= -cv+w$ and $z_{n+1} = u - v$. A Nussbaum function is employed to handle the nonlinearity caused by $\frac{\partial \eta_{sat}(v)}{\partial v}$ which is impossible to be used in cancellation-based design procedures \cite{wen:2011robust}.

\begin{figure}[h]
\begin{center}
\includegraphics[width=0.6\linewidth]{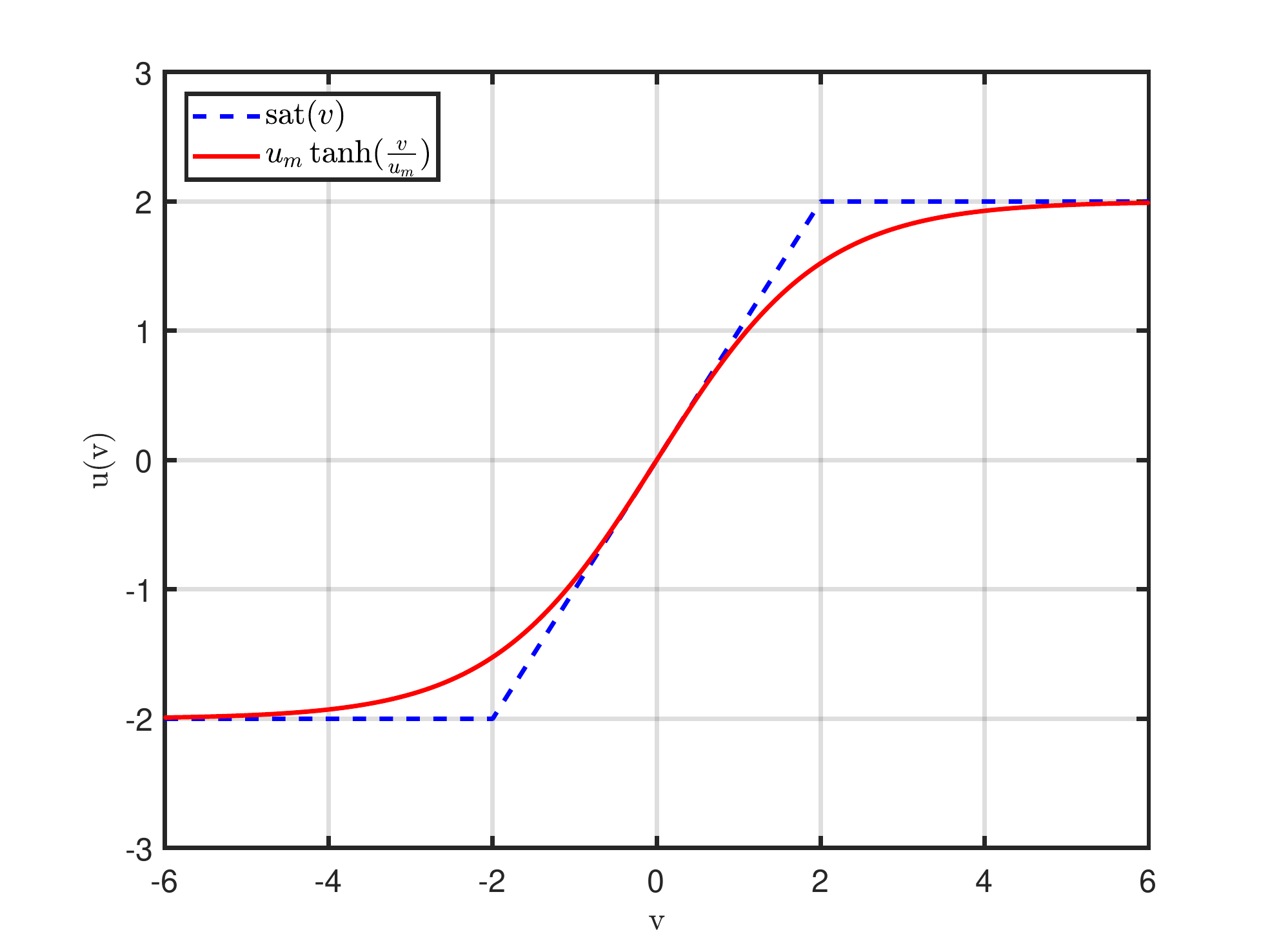}
\caption{Saturation approximation ($u_m = 2$).}
\label{fig:paper13.sat_apprx}	
\end{center}
\end{figure}

If the saturation model is known, an auxiliary design system is introduced to modify the control law from normally backstepping to $v=g_n^{-1}[g_{n-1}z_{n-1}+f_n+\dot{\alpha}_{n-1}- \phi_n^\top\hat{\theta}_n-c_n(z_n-e)]$, with an auxiliary state $e$. The state dynamics $\dot{e}$ is updated with $\Delta u = u-v$ when $e$ is greater than a specific boundary \cite{chen:2011adaptive}.

\subsubsection{Input deadzone}
Deadzone occurs frequently in industrial applications, e.g., gear transmission servo system, DC motor, hydraulic aircraft elevator control system, and valve. It is a memoryless nonlinearity, a non-differential function, and insensitive to a small control input. Consequently, the effect of the deadzone is the undesired chattering, which is a problem in high-precision control. In practice, the deadzone barriers are normally unknown.

There are several methods to handle the deadzone effects, such as deadzone inverse, adaptive control with NN, and robust control with a combination of linear input and bounded disturbance.

A general expression of deadzone operator is presented in Figure~\ref{fig:paper13.deadzone_general} and given by
\begin{equation}
u = \text{Dead}(v)= \begin{cases}
g_r(v), & \text{if } v \geq b_r, \\
0,      & \text{if } b_l < \theta < b_r, \\
g_l(v), & \text{if } v \leq b_l, 
\end{cases}
\end{equation}
where $g_r$ and $g_l$ are the functions in the right and left parts,
$b_l$ and $b_r$ are the barriers in the right and left parts. The challenges of the deadzone nonlinearity is to find the deadzone inverse, i.e., $v=\text{Dead}^{-1}(u)$. The parameters for a deadzone operator are the breaking points $b_r$ and $b_l$, slopes of $g_r$ and $g_l$. Hereafter, $b_r>0$ and $b_l<0$.
The breaking points can be known or unknown, symmetric or asymmetric. The slopes of $g_r$ or $g_l$ can be constant or a function, and they can be the same or not in the positive and negative domains. 
From simple to complex, three types of deadzones are discussed.

\begin{figure}[h]
\begin{center}
\includegraphics[width=0.6\linewidth]{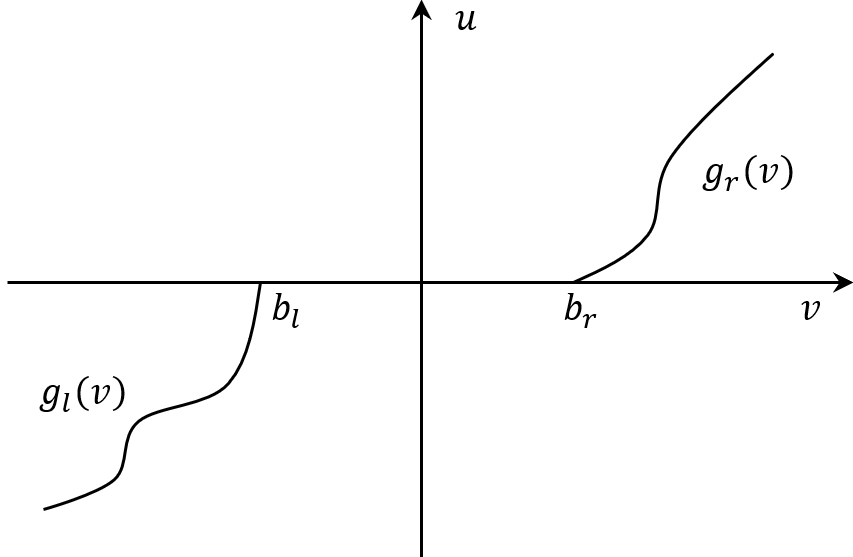}
\caption{General deadzone operator.}
\label{fig:paper13.deadzone_general}	
\end{center}
\end{figure}

We start from the simplest form with linear slopes in both sides, i.e.,
\begin{equation}
u = \text{Dead}(v)= \begin{cases}
k_r (v-b_r), & \text{if } v \geq b_r, \\
0,     			 & \text{if } -b_l < \theta < b_r, \\
k_l (v-b_l), & \text{if } v \leq b_l, 
\end{cases}
\end{equation}
where $k_r$ and $k_l$ are the slopes in the right and left sides, respectively. When all the parameters are known, a direct deadzone inverse \cite{tao:1996adaptive} can be applied, given by
\begin{equation}\label{eq:paper13.deadzone.inverse1}
v = \text{Dead}^{-1}(u)= \begin{cases}
u/k_r + b_r, & \text{if } u \geq 0, \\
0,     			 & if u = 0, \\
u/k_l + b_l, & \text{if } v \leq 0.
\end{cases}
\end{equation}
However, the deadzone inverse (\ref{eq:paper13.deadzone.inverse1}) is not a smooth function. 
The asymmetric deadzone inverse approximation with constant slopes and known width can be $v \approxeq \eta_{dead}(u) = \frac{u + k_r b_r}{k_r}\phi_r + \frac{u + k_l b_l}{k_l}\phi_l$, where $\phi_r = \frac{e^{u/e_0}}{e^{u/e_0}+e^{-u/e_0}}$, $\phi_l = \frac{e^{-u/e_0}}{e^{u/e_0}+e^{-u/e_0}}$, and $e_0$ is a designed parameter \cite{zhou:2006adaptiveDeadzone}. More general deadzone inverse models are tabulated in Table~\ref{tb:paper13.differentialAsymmetricDeadzone} \cite{zuo:2016control}. The approximation error, defined by $d_{dead}(v) = \text{Dead}(v) - \eta_{dead}(v)$, is bounded and can be made arbitrarily small if the parameter $\rho_{dz}$ is sufficiently large, i.e., $\lim_{\rho_{dz} \to \infty} |\eta_{dead}(v)| = 0$. For the second equation in Table~\ref{tb:paper13.differentialAsymmetricDeadzone}, $d_{dead}(v)<k_{dz} \ln\frac{2}{\rho_{dz}}$. However, a large $\rho_{dz}$ results in aggressive close-loop response which may degrade the system.

\begin{table*}[t]
\newcommand{\tabincell}[2]{\begin{tabular}{@{}#1@{}}#2\end{tabular}}
\centering
\caption{Differential asymmetric deadzone approximation model \cite{zuo:2016control}.}
\label{tb:paper13.differentialAsymmetricDeadzone}
\begin{tabular}{cc|l}
\toprule[2pt]
$k_r=k_l=k_{dz}$ & $b_r=b_l=b_{dz}$ & $\label{eq:paper13.operator.zuo2016_deadzone.simplified2}
\eta_{dead}(v) = k_{dz} v + \frac{k_{dz}}{2\rho_{dz}} \ln\frac{\cosh{\rho_{dz} (\theta_{dz} - b_{dz})}}{\cosh{-\rho_{dz} k_{dz} (v + b_{dz})}}$\\
$k_r=k_l=k_{dz}$ & $b_r\neq b_l$ & $\label{eq:paper13.operator.zuo2016_deadzone.simplified1}
\eta_{dead}(v) = k_{dz}\left(v + \frac{b_l-b_r}{2}\right) + \frac{k_{dz}}{2\rho_{dz}} \ln\frac{\cosh{\rho_{dz} (v - b_r)}}{\cosh{-\rho_{dz} k_{dz} (v + b_l)}}$ \\
$k_r\neq k_l$ & $b_r\neq b_l$ &	
$\label{eq:paper13.operator.zuo2016_deadzone.differential}\eta_{dead}(v) = \frac{1}{\rho_{dz}} \ln\frac{1+e^{\rho_{dz} k_r (v - b_r)}}{1+e^{-\rho_{dz} k_l (v + b_l)}}$
\\
\bottomrule[2pt] 
\end{tabular}
\end{table*}

\begin{figure}
\begin{center}
\label{fig:paper13.operator.zuo2016_deadzone}
\includegraphics[width=\linewidth]{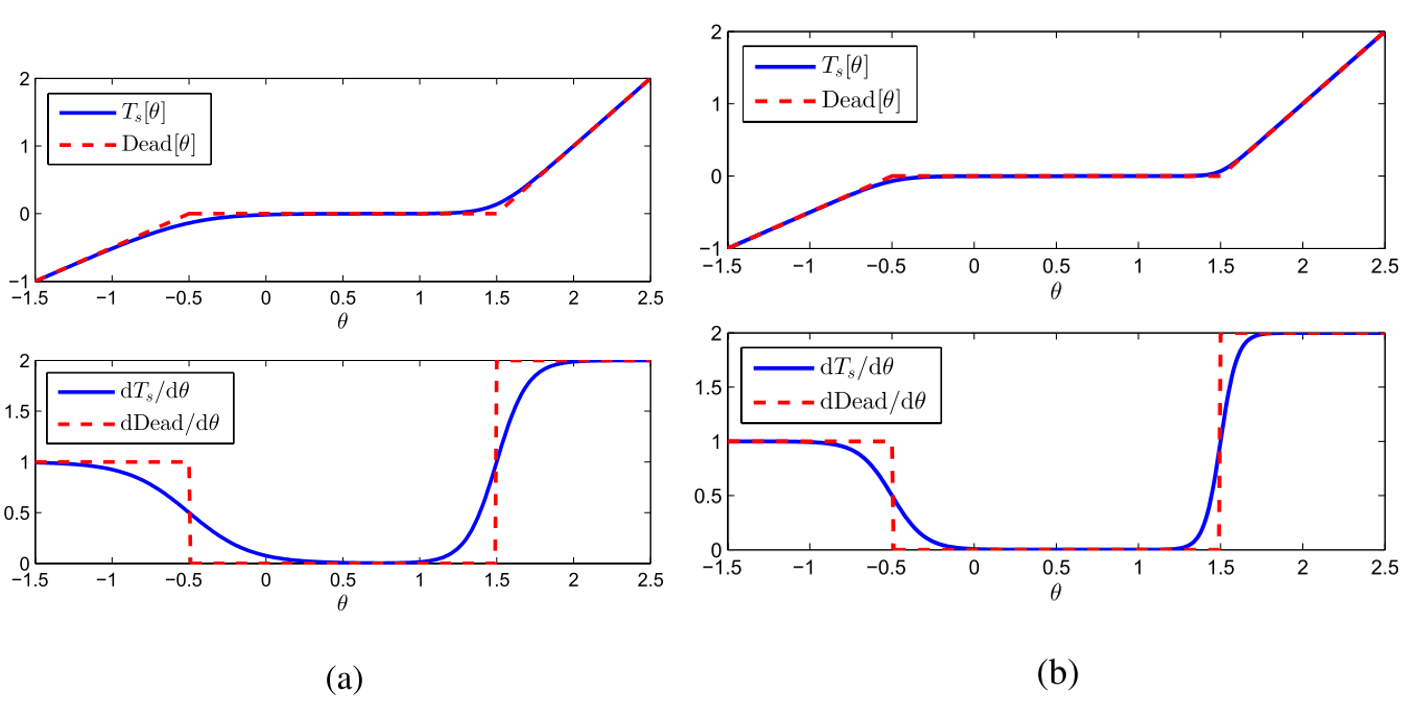}
\caption{Deadzone approximation with different soft degrees: $k_l = 1$, $k_r = 2$, $a_l = 0.5$, and $a_r = 1.5$. (a) $\rho_{dz} = 5$. (b) $\rho_{dz} = 10$. \cite{zuo:2016control}}
\end{center}
\end{figure}

In practical applications, the deadzone breakpoints are always unknown. Early adaptive solutions assume that the slopes and breakpoints are unknown parameters, namely, adaptive deadzone inverse. Adaptive laws are adopted to estimate parameters $\hat{k}_r$, $\hat{k}_l$, $\widehat{k_r b_r}$, and $\widehat{k_l b_l}$ \cite{recker:1991adaptive}. 
The deadzone is divided into separate smooth regions and the LFC is $V = V_n + \tilde{\theta}_r^\top \Gamma_r^{-1} \tilde{\theta}_r + \tilde{\theta}_l^\top \Gamma_l^{-1} \tilde{\theta}_l$, where the adaptive states are $\theta_r = [k_r,k_r b_r]^\top$ and $\theta_l=[k_l, k_l b_l]^\top$, and $\Gamma_r$ and $\Gamma_l$ are positive definite matrices. 

Alternatively, robust adaptive control approaches are proposed using Lemma~\ref{lemma:paper13.boundedV} \cite{wang:2004robust,ibrir:2007adaptive}.
When $k_l$, $k_r$, $b_l$ and $b_r$ are unknown, but stay within known ranges, the deadzone function can be separated into a linear term and a bounded disturbance,
\begin{equation}\label{eq:paper14.deadzone.linear.disturbance}
\text{Dead}(v) = k_{dz} v + d_{dead}(v)
\end{equation}
where 
$d_{dead}(v) = \begin{cases}
-k_r b_r, & \text{if } v \geq b_r, \\
-k_{dz} v,  & \text{if } -b_l < v < b_r, \\
-k_l b_l, & \text{if } v \leq b_l.
\end{cases}$ and
$k_{dz}=\begin{cases}
k_r& \text{if } v \geq b_r\\
k_l& \text{if } v \leq b_l
\end{cases}$. 
Based on the assumption that $k_{dz}\in[k_{\min},k_{\max}]$, $b_r\in[b_{r\min},b_{r\max}]$, and $b_l\in[b_{l\min},b_{l\max}]$, we have $|d_{dead}|\leq\max\{k_{\max} b_{r\max},\allowbreak-k_{\max} b_{l\max} \}$. Besides robust adaptive control, the disturbance $d_{dead}$ can be handled with estimation-based methods.

A more complicated case is the unknown nonlinear $g_r$ and $g_l$. 
When the deadzone width parameters $b_r$ and $b_l$ are bounded constants, $g_r$ and $g_l$ are smooth functions with bounded slopes, a separate model \cite{meng:2014adaptive} similar to (\ref{eq:paper14.deadzone.linear.disturbance}) is given by $\text{Dead}(v) = K^\top \Phi v + d_{dead}$ where the denotations of $K$, $\Phi$, and $d$ can be found in \cite{zhang:2008adaptive}. 
Since $d_{dead}$ is bounded, the deadzone is transferred to the problems relevant to bounded disturbance with unknown time-varying control coefficient. NN/FLS-based approximation approaches are used to estimate the deadzone dynamics with unknown $g_r$ and $g_l$, as well as unknown $b_r$ and $b_l$ \cite{selmic:2000deadzone}. The deadzone inverse is $v=\text{Dead}^{-1}(u) = u + u_{NN}$, where 
$$u_{NN} = \begin{cases}
g_l^{-1}(u), & \text{if } u<0\\
0, & \text{if } u=0 \\
g_r^{-1}(u), & \text{if } u>0
\end{cases}.$$

\subsubsection{Time-varying control coefficient}\label{sec.paper13.unknownCtrlCoeff}
If the control coefficients $g_i$ is an unknown nonzero constant with known sign, $V_i = \frac{1}{g_i}z_i$ can be used as the LFCs. In $\dot{V}_i$, the resulting term $\frac{1}{g_i} f_i$ is similar to a typical adaptive backstepping design, where $\theta_{gi}=\frac{1}{g_i}$ \cite{ge:2005robust,wang2020adaptive}. Singularity problem can be avoided with this modifications. 

A more challenging scenario is that $g_i$ is only partly known, e.g., $\dot{x}_i=f_i+(g_i + \Delta g_i) x_{i+1}$ where $\Delta g_i$ is the unknown control coefficients. The basic assumption is that the uncertain parts $\Delta g_i$ is bounded, i.e., there exist positive constants $\overline{g}_{i\Delta}$ s.t. $|\Delta g_i(\bar{x}_i)|\leq \overline{g}_{i\Delta}$ \cite{chen:2011adaptive}. The aim is to design virtual control laws in case of the singularity issues caused by the $g_i$ part. In addition, this assumption also ensures the sign of $g_i+\Delta g_i$ is unchanged. The virtual control law is modified to $ \alpha_1 = (g_1 + \gamma_1)^{-1} (-c_1 z_1 - f_1 + \dot{x}_{1d})$, where $\gamma_1>\overline{g}_1$ is selected to avoid the singularity problem of $(g_1 + \gamma_1 )^{-1}$. 

To avoid the singularity problem when $g_i(\bar{x}_i(t))$ is unknown, integral Lyapunov functionals can be used, such as
\begin{itemize}
\item $V_{1,I} = \int_{0}^{z_1} \sigma \beta_i(\sigma+x_{1d}) d\sigma$
\item $\dot{V}_{1,I} = z_1 \beta_1 \dot{z}_1+
\int_{0}^{z_1} \sigma \frac{\partial \beta_i(\sigma+x_{1d})}{\partial x_{1d}} \dot{x}_{1d} d\sigma$
\item $V_{i,I} = \int_{0}^{z_i} \sigma \beta_i(\bar{x}_{i-1},\sigma+\alpha_{i-1}) d\sigma, \ i=2,\cdots,n$
\end{itemize}
where $\sigma = \theta z_1$ and $\beta_i = \frac{\overline{g}_i}{|g_i|}$. The integral Lyapunov functionals is applicable to autonomous systems since $\frac{1}{2}z_i^2 \leq V_{i,I} \leq \frac{z_i^2}{\underline{g}_i} \int_{0}^{1} \theta \overline{g}(\theta z_i + \alpha_{i-1})$ \cite{ge:2004adaptive}.
The design continues using the Nussbaum function and NN/FLS.

For the cases, the control coefficients $g_i$ is non-positive or non-negative. 
For example, absolute is a quite commonly used operator in hydraulic systems. When the absolute operator appears in $g_i(x)$ in Eq.~(\ref{eq:paper13.strictFeedbackForm}), $g_i$ is no longer a first-order continuous function. Consequently, a singularity problem is aroused, i.e., $1/g_i \to \infty$ when $g \to 0$. Sign function is usually applied to remove the absolute operator, $|a|=sgn(a)a$. An approximation $sgn(a) \approxeq tanh(a/\delta)$ transfers $g(x)$ into a continuously Lipschitz function where $\delta$ is a small enough positive constant \cite{yao:1998nonlinear,zhou:2006adaptiveSaturation}. Another way to avoid $1/g_i$ a singularity problem is to use $g_{i,new} = \frac{g_i^2}{g_i+\varepsilon_i}$ or $g_{i,new} = g_i + \varepsilon_i$ during the control design, where $\delta$ is a small value.

A more general problem is that the sign of $g_i$ is unknown. Nussbaum function presented in Section~\ref{sec.paper13.NussbaumFunction} is widely used to handle such an issue.

\subsubsection{Actuator failure}
Actuator failures change the output and parameters, introduce additional system uncertainties and disturbances, and result in performance deterioration and even accidents. A failure is normally uncertain in time and often unrecoverable. Modern researches on fault-tolerant control are in two directions. (i) Firstly, the model-based redundancy approach for fault-tolerant control is based on a bank of residual signals generated by multiple online monitoring modules running in parallel with specific possible failures. If the failures are not contained in the bank, the performance is unreliable. (ii) Another direction is the adaptive failure compensation design without explicit failure detection, which has a simpler structure. The objective is to compensate for the effects of reasoning from the actuator failures, and meanwhile, to ensure the asymptotic tracking performance with a bounded error. The controller remains the same structure through the running; therefore, it is easier to implement. In this survey, we focus on (ii) adaptive failure compensation design.

Failures are categorized into the total loss of effectiveness (TLOE) and partial loss of effectiveness (PLOE). 
Suppose that there are in total $m$ actuators in the system, i.e., $u=[u_1,\cdots,u_m]^\top$
For a TLOE scenario, redundant actuators are required. A common assumption is that the remaining actuators are fully actuated, and the desired control objective is still achievable for up to $m-1$ actuator faults for an SISO system~\cite{tang:2003adaptive,li:2016adaptiveQuantization}. This assumption ensures the controllability of the plant with the remaining actuation power and the existence of a normal solution for the actuator failure compensation problem.

In literature, there are three types of actuator failure models. From simple to complex, they are listed as follows.
\begin{itemize}
\item  If failure for the $j^{th}$ actuator occurs at $t_j$, i.e.,
\begin{equation}
u_j(t)=\bar{u}_j,\ \forall t\geq t_j,\ j = 1,\cdots,m,
\end{equation}
where $\bar{u}_j$ and $t_j$ are unknown. The widely-used actuator fault model is
$u=\sigma \bar{u} + (I_m - \sigma) v$, where $\sigma = diag\{\sigma_1,\cdots,\sigma_m\}$ is a matrix containing failure patterns $\sigma_j$ and $\sigma_j=\begin{cases}
1 & \text{if the }j^{th}\text{ actuator fails},\\
0 & \text{otherwise}.\end{cases}$
Then $\sigma \bar{u}$ and $(I_m - \sigma)$ can be estimated using adaptive laws
\cite{tao:2001adaptive}.

\item Actuator model with both gain fault and bias fault is modeled by \cite{shen:2014adaptive,chen:2016adaptive}, i.e.,
\begin{equation}
u_j = \rho_j v_{j} + b_{uj}.
\end{equation}
The model has four modes as follows.
\begin{itemize}
\item Failure-free: $\rho_j = 1$ and $b_{uj}=0$;
\item PLOE: $\rho_j \in (0,1)$ and $b_{uj}=0$;
\item TLOE: $\rho_j = 0$ and $b_{uj} = 0$;
\item Bias fault: $\rho_j = 0$ and $b_{uj} \neq 0$.
\end{itemize}

\item First- \cite{boskovic2003robust}, second-\cite{boskovic2005adaptive}, and third-order \cite{boskovic:2010decentralized} dynamic actuator failure models are proposed. The first-order model is given by
\begin{equation}
\dot{u}_j = -(1-\sigma_j) \lambda_j (u_j - k_j v_j)
\end{equation} and the second-order model is given by
\begin{equation}
\begin{aligned}
\dot{u}_{1j} &= u_{2j}, \\
\dot{u}_{2j} &= -(\lambda_{2j}+\sigma_j\beta_j)u_{2j}+(1-\sigma_j)\lambda_{1j} (k_j v_j - u_{1j})
\end{aligned}
\end{equation}
where
$\lambda_{j}\gg 1$, $\lambda_{1j}\gg 1$, $\lambda_{1j}\gg\lambda_{2j}$, and $\lambda_{2j}+\beta_j \gg 1$. the three modes are
\begin{itemize}
\item Failure-free: $\sigma_j = 0$ and $k_j = 1$;
\item PLOE: $\sigma_j = 0$ and $k_j\in(0,1)$;
\item TLOE: $\sigma_j = 1$.
\end{itemize}
\end{itemize}

For a PLOE issue, the control coefficient can be considered as time-varying unknown parameters, which can be solved by NN-based design \cite{jiang:2014adaptive} and Nussbaum function \cite{shen:2017fault}. The fault detection and isolation mechanism can be integrated into the adaptive control to improve the performance, e.g., the deactivation function for failed actuators is realizable. NNs can be adopted to optimize the effects of the failure after fault detection but before the isolation \cite{zhang:2004adaptive}. For the system with healthy actuators for backups, an adaptive state-feedback fault-tolerant control function can be achieved by combining the BLF and a monitoring function \cite{ouyang:2017adaptive}.

To ensure the transient performance of the tracking error, prescribed performance bounds based control design is used \cite{wang:2010adaptive}. The tracking error should stay in a tube of a prescribed decreasing smooth prescribed performance function $\eta(t)$, i.e., $-\underline{\delta} \eta(t) < z_1(t) < \overline{\delta} \eta(t)$, where $0<\underline{\delta}<\overline{\delta}\leq 1$ are given positive constants. A new error state is defined as $z_1 = S^{-1}(\frac{z_1}{\eta})$, where the smooth, strictly increasing, and invertible function $S(z_1)$ satisfies: (i) $\underline{\delta} < S(z_1) < \overline{\delta}$, (ii) $\lim\limits_{z_1\to+\infty}S(z_1) = \overline{\delta}$ and $\lim\limits_{z_1\to-\infty}S(z_1) = \underline{\delta}$, and (iii) $S(0)=0$.
For example, functions $\eta(t)$ and $S$ are designed as $\eta(t)=[\eta(0)-\eta(\infty)]\exp(-at) + \eta(\infty)$ and $S(v)=\frac{\overline{\delta}e^{(v+r)}-\underline{\delta}e^{-(v+r)}}{e^{(v+r)} + e^{-(v+r)}}$, where $a,\eta(0),\eta(\infty)>0$, $\eta(0)>\eta(\infty)$, and $r=\ln(\underline{\delta}/\overline{\delta})/2$.

%



\subsection{Time-delay effects}
Widely existing in chemical systems, biological systems, economic systems, and hydraulic/pneumatic systems, the time-delay phenomenon usually deteriorates the performance of a closed-loop system. 
A system stabilized by a feedback control law may become unstable after introducing the time-delay effects. 
Time-delay effects probably exist everywhere in the system \eqref{eq:paper13.strictFeedbackForm}, e.g., $f_i$, $x_{i+1}$, $u$, and $y$. The delayed times are often unknown, which can be a constant value for all parameters, constant but different values for different parameters, and time-varying. 
The feedback laws have to be modified to adjust the time-delay effects. 
There exist two main types of Lyapunov functions for a time-delay system, i.e., Lyapunov-Razumikhin approach \cite{jankovic:2001control} and Lyapunov-Krasovskii approach \cite{malisoff:2009constructions}.  Most studies on adaptive backstepping adopt Lyapunov-Krasovskii functionals, and only a few studies rely on Razumikhin lemma \cite{hua:2008robust}. Lyapunov-Krasovskii approach is a predictor-like technique according to the Lyapunov-Krasovskii Theorem.


\subsubsection{State delay}
A strict-feedback system with state-delay is given by
\begin{subequations}
\begin{align}
\dot{x}_i(t) &= f_i(\bar{x}_i(t)) + f_{di}(\bar{x}_i(t-\tau_{di})) + g_i(\bar{x}_i(t))x_{i+1}(t), \\
\dot{x}_n(t) &= f_n(\bar{x}_n(t)) + f_{dn}(\bar{x}_n(t-\tau_{dn})) + g_n(\bar{x}_n(t)) u(t),
\end{align}
\end{subequations}
where 
$\tau_{di}>0$ denotes the delayed time, 
$f_{di}(\bar{x}_i(t-\tau_{di}))$ are the time-delay terms, and $\bar{x}_i(t-\tau_{di}) = [x_1(t-\tau_{d1}),x_2(t-\tau_{d2}),\cdots,x_i(t-\tau_{di})]^\top$.

A common assumption for a time-delay system is that the unknown time delays are bounded by a known constant, i.e., $\tau_{di}\leq \tau_{d,max}$. At the beginning, strict-feedback form with known equal delay $\tau_{d1}=\cdots=\tau_{dn}=\tau_d$ is handled by robustness-based approach by assuming all functions are known and bounded \cite{nguang:2000robust,ho:2005adaptive}. However, the delays in different states always vary in practical applications, i.e., $\tau_{d1}\neq \cdots \neq \tau_{dn}$ \cite{ge:2005robust,ge:2007approximation,chen:2009novel}. 

In addition, various assumptions of bounded parametric time-delayed terms are adopted to make $f_{di}$ to be compensable. The absolute value of the time-delay term is bounded by known smooth functions $\rho_i(\bar{x}_i)$ in several parameter-separation forms. Commonly used assumptions are listed as follows:
\begin{itemize}
\item $|f_{di}(\bar{x}_i(t-\tau_{di}))| \leq \sum\limits_{j=1}^{n}\rho_j(\bar{x}_j)$ \cite{ge:2007approximation},
\item $|f_{di}(\bar{x}_i(t-\tau_{di}))| \leq \rho_i(\bar{x}_i(t-\tau_{di}))$ \cite{ge:2004adaptive},
\item $|f_{di}(\bar{z}_i(t-\tau_{di}))|\leq \sum\limits_{j=1}^{i}|z_j(t-\tau_{dj})|\rho_{ij}(\bar{z}_j(t-\tau_{dj}))$ \cite{nguang:2000robust,ho:2005adaptive},
\item $|f_{di}(\bar{x}_i(t-\tau_{di}))| = \theta_{di}^\top \phi_{di}(\bar{x}_i(t-\tau_{di})) + \delta_{di}(\bar{x}_i(t-\tau_{di}))$ \cite{ge:2005robust,zhang:2015exact},
\end{itemize}
where $\rho_{ij}(\cdot)$ is a known continuous and smooth function, $\theta_{di}\in\mathbb{R}^{n_i}$, $\phi_{di}:\mathbb{R}^{i}\to\mathbb{R}^{n_i}$ is a known smooth function vector, $\delta_{di}$ is a bounded unknown smooth function, i.e., $|\delta_{di}(\bar{x}_i(t-\tau_{di}))| \leq c_{di}\rho_i(\bar{x}_i(t-\tau_{di}))$ where $c_{di}$ is an unknown constant.

An additional Lyapunov-Krasovskii functional is added to the LFC in the corresponding step. 
A basic integrate-type Lyapunov-Krasovskii functional and its derivative are 
\begin{subequations}
\begin{align}
V_{i,LK} =& \int_{t-\tau_d}^{t}S_i(\bar{z}_i(\sigma))d\sigma,\\
\dot{V}_{i,LK} =& S_i(t) - S_i(t-\tau_d),
\end{align}
\end{subequations}
where $S_i(\bar{z}_i(t))$ is a positive definite function, e.g., $S_i(\bar{z}_i(t)) = \rho_i^2(\bar{z}_i(t))$. The main idea of Lyapunov-Krasovskii functional is to cancel the delayed term with the selection of functions $S_i$. A virtual control law only compensates the terms without delay, and the delayed terms are canceled by a selection of function $S_i$. Depending on the number of time delayed components in the $i^{th}$ step, the general Lyapunov-Krasovskii functional is selected as $V_{i,LK} = \sum_{\mathcal{J}}\int_{t-\tau_{dj}}^{t}S_j(\bar{z}_j(\sigma))d\sigma$ when several delays $\tau_{dj}$ exist in the a same dynamics $\dot{x}_i$, with $\mathcal{J}$ is the set containing all indices of the delayed states~\cite{zhang:2015exact}. When multiple time delays exist in the states $x_1,x_2,\cdots,x_n$, a nonlinear decomposition strategy is adopted, e.g., $f_{di}(\bar{x}_i(t-\tau_{di})) \leq \sum_{j=1}^{i} \beta_{ij} (x_j(t-\tau_{di}))$. Then, the time derivative of $V_{i,QF}$ has a part $z_i f_{di}$, which can be separated by the Young's inequality, i.e., $z_i f_{di} \leq \frac{i}{2}z_i^2 + \frac{1}{2} \sum_{j=1}^{i} \beta_{ij}^2 (x_j(t-\tau_{dj}))$. Lyapunov-Krasovskii functionals are designed to compensate $\frac{1}{2} \sum_{j=1}^{i} \beta_{ij}^2 (x_j(t-\tau_{dj}))$ \cite{niu:2017adaptive}, such as $V_{i,LK} = \frac{1}{2}\int_{t-\tau_{di}}^{t} \sum_{j=1}^{i} \beta_{ij}^2 (x_j(\sigma)) d\sigma$. The effects of the additional Lyapunov-Krasovskii functional is an addition in the the control gain $c_i$.

Besides, the delayed function $f_{di}$ can be approximated and canceled using NNs/FLSs \cite{ge:2007approximation,chen:2009novel}. However, such approaches are limited when the nonlinearities change rapidly.

\subsubsection{Delayed integrator}
A system with delayed integrator is given by
\begin{subequations}
\begin{align}
\dot{x}_i &= f_i(\bar{x}_i) + g_i(\bar{x}_i)x_{i+1}(t-\tau_{di}), \\
\dot{x}_n &= f_n(\bar{x}_n) + g_n(\bar{x}_n) u.
\end{align}
\end{subequations}

A chain of delayed integrators (i.e., $f_i = 0$) \cite{zhou:2012stabilization} and linear system (i.e., $f_i = \sum\limits_{j=1}^{i} a_{ij} x_j$) \cite{bekiaris:2010stabilization} have been studied. However, the control for nonlinear systems with unknown delayed integrators is still open.

\subsubsection{Input delay}
There are several methods to handle a system with delayed input, which is given by
\begin{subequations}
\begin{align}
\dot{x}_i &= f_i(\bar{x}_i) + g_i(\bar{x}_i)x_{i+1}, \\
\dot{x}_n &= f_n(\bar{x}_n) + g_n(\bar{x}_n) u(t-\tau_d).
\end{align}
\end{subequations}

First, pade approximation approach can be used to handle time-varying input delay \cite{khanesar:2015adaptive,li:2017adaptive}. One additional state $x_{n+1}$ and a lowpass process is introduced to cancel the influence of the delayed time, i.e.,
\begin{subequations}
\begin{align}
\dot{x}_i &= f_i(\bar{x}_i) + g_i(\bar{x}_i)x_{i+1}, \\
\dot{x}_n &= f_n(\bar{x}_n) + g_n(\bar{x}_n) (x_{n+1}-u),\\
\dot{x}_{n+1} &= -\gamma_d x_{n+1} + 2\gamma_d u,
\end{align}
\end{subequations}
where $\gamma_d=\frac{2}{\tau_d}$. The last error state is defined as $z_n := x_n-\alpha_{n-1}+\frac{1}{\gamma_d}x_{n+1}$. When $\tau_d$ is unknown, $\gamma = \gamma_n + \tilde{\gamma}$, where $\gamma_n$ is assumed to be known and $\tilde{\gamma}$ is the bounded error. 
The Pade approximation approach is only capable to work in a short delay, but not a long delay. 

Moreover, emulation is another approach for the input delay with a short delay. The controller is first designed without delay, then the upper bound of the delay is found such that the closed-loop system is still asymptotically stable \cite{mazenc:2006backstepping}.

Besides, the unknown arbitrarily long actuator delay length can be handled by prediction-based boundary control \cite{krstic:2008backstepping} and predictor feedback \cite{bresch:2009adaptive}, particularly used on linear time-invariant, finite-dimensional, and completely controllable system. The effects of the delay is compensated with an integration over the delay period. 
The delay time is estimated by a time-delay identifier when using the prediction-based boundary control.
The main idea of the predictor feedback approach is to modeled the actuator time delay effects as a transport partial differential equation \cite{bresch:2012adaptive}. Assuming the distributed input $u(x,t) = U(t+\tau_d(x-1))$, the delayed input is replace by a real-time input $v$, s.t.,
\begin{subequations}
\begin{align}
\tau_d \frac{\partial u}{\partial t}(x,t)& = \frac{\partial u}{\partial x}(x,t),\  x\in[0,1),\\
u(1,t) &= U(t). 
\end{align}
\end{subequations}
However, the distributed terms may not always be easy to compute. Furthermore, this approach is not applicable to nonlinear system due to the inconvenience integration over the delay interval. 

An alternative way to solve a system with an arbitrarily long input delay is proposed in \cite{mazenc:2016new} without distributed terms. A time-delay system is transformed into another dynamics with a delayed system state as its input \cite{zhou:2009adaptive}.


\subsubsection{Measurement delay}
In additional to the aforementioned time-delayed systems, the measurement delay is a few less popular issue, i.e., $y=x_1(t-\tau_d)$, which can be transformed to a state-delay problem. The zero-order hold delayed measurement \cite{karafyllis:2012nonlinear} is overcome by predictor-type designs \cite{karafyllis:2013stabilization} and estimators based on FLS \cite{wang:2015adaptive}.

\subsection{Pure feedback system}
A limitation of backstepping approach is that it requires the system to be in a strict-feedback form. A milder state equation, namely nonlinear pure-feedback system, is given by
\begin{subequations}\label{eq:paper13.pureFeedback}
\begin{align}
\dot{x}_i &= f_i(\bar{x}_{i},x_{i+1}), \ i\in\mathcal{I}, \\
\dot{x}_n &= f_n(\bar{x}_n,u).
\end{align}
\end{subequations}
Comparing to the strict-feedback system in Eq.~(\ref{eq:paper13.strictFeedbackForm}), a pure-feedback system (\ref{eq:paper13.pureFeedback}) has no affine appearance of the variables to be used as virtual control. The basis of the following solutions is to receive affine variables used as virtual control.

\subsubsection{Taylor series expansion}
If the pure-feedback system has a strong relative degree and can be transformed into an integrator chain, 
Taylor series expansion can be adopted to design an observer and neural adaptive controller \cite{leu:2005observer}. The high-order terms are modeled as disturbances. 

\subsubsection{Mean value theorem}
Mean value theorem is the most widely used method. A pure-feedback system (\ref{eq:paper13.pureFeedback}) is transformed to a strict-feedback form with unknown control coefficient \cite{ge:2002adaptive}, i.e., 
\begin{subequations}
\begin{align}
f_i(\bar{x}_{i+1})&=f_i(\bar{x}_i,0)+\mu_i(\bar{x}_i,x_{i+1}^0)x_{i+1}, \\
f_n(x,u) &= f_n(\bar{x}_n,0) + \mu_n(\bar{x}_n,u^0)u,
\end{align}
\end{subequations}
where $\mu_i(\bar{x}_{i+1}) = \frac{\partial f_i(\bar{x}_{i+1})}{\partial x_{i+1}}$ and
$\mu_n(\bar{x}_n,u) = \frac{\partial f_n(x,u)}{\partial u}$ exist but are unknown,
$x_{i+1}^0$ is an unknown point between 0 and $x_{i+1}$, and 
$u^0$ is an unknown value between 0 and $u$. Hence, the pure-feedback systems are transferred into a strict-feedback system with unknown functions $f_i$ and control coefficients $\mu_i$. Refer to the above-mentioned methods, the deduction is based on assumptions with bounded control gains, i.e., $\underline{\mu}_i<\mu_i<\overline{\mu}_i$. The error dynamics is then transferred to 
\begin{equation}\label{eq:paper13.pureFeedback.meanValueTheorem}
\dot{z}_i  = f_i(\bar{x}_i,\alpha_i^*)+ g_{i\mu_i}(x_{i+1}-\alpha_i^*),
\end{equation}
where $g_{i\mu_1}=g_i (\bar{x}_i,\mu_i x_{i+1} + (1-\mu_i)\alpha_i^*)$, with $\mu_i\in(0,1)$,
$\alpha_i^*$ is the ideal control input such that $f_1(x_1,\alpha_1^*)-\dot{x}_{1d}+k_1z_1 = 0$. Assuming that $|g_{1\mu_1}|$ is bounded within a known constant, and $\alpha_1^*$ is a function of $x_1$ and $x_{1d}$. 

First, the system (\ref{eq:paper13.pureFeedback.meanValueTheorem}) can be solved using NNs to approximate $\alpha_i^*$ \cite{ge:2004adaptive}. Semiglobal stability can be proved when the NN stays in some compact sets which can be made arbitrarily large with a sufficiently large number of NN nodes. 
Alternatively, input-to-state stability (ISS) is adopted to divide the entire interconnected system into two subsystems with overall stability guaranteed by the small-gain theorem \cite{wang:2006iss}. It is proved that the first error dynamic system with state $z$, inputs $\text{col}(\tilde{W}_1,\tilde{W}_2,\cdots,\tilde{W}_n)$ and $\text{col}(\varepsilon_1,\varepsilon_2,\cdots,\varepsilon_n)$ is ISS. Furthermore, the NN weight estimator with state $\text{col}(\tilde{W}_1,\tilde{W}_2,\cdots,\tilde{W}_n)$, inputs $z$ and $\text{col}(\tilde{W}_1^*,\tilde{W}_2^*,\cdots,\tilde{W}_n^*)$ is ISS. Then the stability of the interconnected system is guaranteed.

Instead of approximation on all unknown, robustness-based methods can be developed by assumed bounded uncertainties, e.g., $f_i(\bar{x}_i,\alpha_i^*) \leq \sigma_i\rho_i(\bar{x}_i) \sum\limits_{j=1}^{i}|x_j|$. A high-gain idea is used with dynamic control gains $\dot{c}_i$, and the final stability is proved with Barbalat's lemma \cite{hou:2017adaptive}. However, a rapid oscillation has been noticed in the control input.
Another method with simplified control law is  the prescribed performance control \cite{bechlioulis:2014low}. The system is bounded by predefined trajectories, where the convergence rate and exhibiting maximum overshoot ensure preassigned performance. Without the attempts to cancel all nonlinearity, the control law is similar to that using BLF, i.e., $u=-c_n T_f(\frac{x_n - \alpha_{n-1}}{\rho_n(t)})$ and $\alpha_i = -c_i T_f(\frac{x_i - \alpha_{i-1}}{\rho_i(t)})$ where $\rho_i$ is a prescribed performance function and $T_f(a) = \ln \frac{1+a}{1-a}$ with $a\in(0,1)$. The control gains are enlarged near the performance boundary.

Singular perturbation theory is combined with backstepping design \cite{hovakimyan:2007dynamic,yoo:2012adaptive}. The control signal is sought as a solution to the fast dynamics and is shown to asymptotically stabilize the original non-affine system.

When using mean value theorem, the main drawbacks are the requirement of positive $\frac{\partial f_i(\bar{x_i},x_{i+1})}{\partial x_{i+1}}$ and the complexity for high-order system. DSC and minimal learning parameter techniques are employed to release the requirement and achieve neutral adaptive control to the system with nondifferential $f_i$ \cite{liu:2016adaptive}. To enhance the estimation performance, the increasing number of NN nodes results in a significant number of adjustable parameters. The online learning time can be very long. A possible solution to reduce the adjustable parameters is to estimate $|W_i^*|^2$ with a single state $\theta$. Another shortage is the circular construction problem, since the NN is used to approximate $u$ $\dot{u}$, and the desired virtual and practical controls $\alpha_i^*$ and $u^*$, simultaneously \cite{ge:2002adaptive}. 

\subsubsection{Brunovsky norm form}
System (\ref{eq:paper13.pureFeedback}) is transferred to an affine strict-feedback form in Brunovsky norm form by introduce low-pass process $\dot{u} = -a u +v$, where $a>0$ \cite{meng:2014adaptive}. Defining states $s_1 := x_1=b_1(x_1)$ and $s_i = \dot{s}_{i-1}=:b_i(\bar{x}_i)$, $i=2,\cdots,n+1$ yields, $\dot{s}_{i}  = b_{i+1}=  \sum\limits_{j=1}^{i}\frac{\partial b_j(\bar{x_j})}{\partial x_j} \dot{x}_j = \sum\limits_{j=1}^{i-1}\frac{\partial b_j(\bar{x}_j)}{\partial x_j} f_j(\bar{x}_{j-1},x_j) 
+ \frac{\partial b_{i}(\bar{x}_{i})}{\partial x_{i}} f_{i}(\bar{x}_{i-1},x_{i})$ and $\frac{\partial b_{i+1}(\bar{x}_{i+1})}{\partial x_{i+1}} = \frac{\partial b_i(\bar{x}_i)}{\partial x_i} \frac{\partial f_{i+1}(\bar{x}_{i},x_{i+1}) }{\partial x_{i+1}} = \prod_{j=1}^{i} g_j(\bar{x}_{j-1},x_j)$
where $g_i = \frac{\partial f_{i+1}(\bar{x}_{i},x_{i+1}) }{\partial x_{i+1}}$ and $g_n(\bar{x}_n,u) = -a\prod_{j=1}^{n} g_j(\bar{x}_{j-1},x_j)$ satisfy the bounded $g_i$ assumption. Consequently, $	\dot{s} = As + B [f_n(\bar{x}_n,u) + g_n(\bar{x}_n,u)v]$, where $s=[s_1,\cdots,s_{n+1}]^\top$, $A= \begin{bmatrix}
0_{n\times 1} &  I_n \\
0 & 0_{1 \times n}
\end{bmatrix}$, $B=[0_{1\times n },1]^\top$, and $C=[1,0_{1\times n }]$. The final control laws have simpler forms if the linear systems are used.

\subsection{Event-triggered systems}
In typical backstepping designs, the continuous control is adopted in the design process and the system updates at discrete-time triggered instants $t_k$, namely, time-triggered system. 
In contrast, the event-triggered control, or event-based system, denotes a system updates aperiodically with the pre-designed event-triggered condition. 
The execution of control tasks update after the occurrence of an event, and the control input (or measurements) is held between two consecutive updates (zero-order hold).
Normally, the monotonically increasing subsequence $t_k$ starts from $t_0 = 0$.
It is a more natural sampling way that of a human controller.
The strategy takes advantage of reducing network traffic loads and improving resource utilization with minor control performance degradation. The event-triggered problem can be categorized according to which part is event-triggered. 

(i) A major research focuses on the event-triggered control input. We assume constant control input between two triggered instants, i.e., $u(t)=v(t_k)$ for all $t\in[t_k,t_{k+1})$. The design procedure is the same as the time-triggered execution, except the last step. 

The event-triggered conditions are designed as the specific error $e_{et}:= u(t)-v(t_k)$ is larger than a preset threshold function. 
The thresholds can be fixed and state-dependent; listed as follows.
\begin{itemize}
\item Fixed threshold strategy \cite{xing2016event,wang2020adaptive}: 
\begin{equation}\label{eq:ETC-threshold_1}
t_{k+1} = \inf\{ t>t_k | |e_{et}(t)|> \delta_{et}\}
\end{equation}

\item State-dependent threshold 1 \cite{tabuada2007event}:
\begin{equation}\label{eq:ETC-threshold_2}
t_{k+1} = \inf\{ t>t_k | |e_{et}(t)|> \kappa_{et} |u(t)| \}
\end{equation}

\item State-dependent threshold 2 \cite{zhang2018event}:
\begin{equation}\label{eq:ETC-threshold_3}
t_{k+1} = \inf\{ t>t_k | |e_{et}(t)|> \kappa_{et} |u(t)| + \delta_{et}\}
\end{equation}
\end{itemize}
where $\kappa_{et}\in (0,1)$ and $\delta_{et}>0$.

Removing the absolute operators and substituting them into $\dot{V}_n$, $v(t)$ can be designed to be a function of the last virtual control $\alpha_i$.
For example, the control input under event-triggered condition \eqref{eq:ETC-threshold_3} becomes $u(t)=\frac{v(t)}{1+\eta_{et,1}(t)\kappa_{et}}-\frac{\eta_{et,2}(t)\delta_{et}}{1+\eta_{et,1}(t)\kappa_{et}}$ where $|\eta_{et,1}|\leq 1$ and $|\eta_{et,2}|\leq 1$. Therefore, the an unknown bounded external disturbance $-\frac{\eta_2(t)\delta_{et}}{1+\eta_1(t)\kappa_{et}}$ is introduced to state equation $\dot{x}_n=f_n(\bar{x}_n) + g_n(\bar{x}_n) u$. The problem can be easily solved by the $\tanh(\cdot)$ function \cite{zhang2018event}.

(ii) Alternatively, if the state-measurements are event-triggered, namely event-sampled, the measurements are considered as a jump, i.e., $\hat{x}_i(t) = x_i(t_k)$, for all $t\in[t_k,t_{k+1})$. 
The error in the event-triggered conditions then becomes $e_{et,i}(t) := x_i-\hat{x}_i$ for event-triggered sampling. 
The jump can be estimated with the newly-defined errors $e_{z_i}=z_i-\hat{z}_i$ \cite{li2017model}. 
To solve the problem by observer-based approach, the weight updates of the NN and FLS between the trigger instants are zero, i.e., $\dot{\hat{W}}_i = 0$, for all $t\in(t_k,t_{k+1})$ \cite{cao2018adaptive}.

The nonlinearity introduced by the event-triggered mechanism is difficult to be fully canceled; hence, only boundness is guaranteed in most studies. The control should satisfy ISS conditions which are difficult to check concerning the error. 
The high control input $u$ results in longer sampling intervals and degrades the control performance.
To improve the control precisely, a switching threshold strategy is proposed based on the amplitude of tracking error \cite{xing2016event}.
The tracking error converges to a domain. Exponential convergence and finite-time stability are studied in \cite{wang2018exponential}.

Though the bandwidth is reduced for the input and state measurements, continuously monitoring is required by the trigger mechanism, and it is impossible for digital platforms. To further release the hardware resources, periodic event-triggered control at fixed equidistant time instants is proposed, in which the event-triggered function is evaluated periodically. Borgers et. al.\cite{borgers2018periodic} redesigns the event function for an available continuous ETC using convex overapproximation techniques.

\subsection{Stochastic system}
A practical system is normally influenced by the external environments; hence, stochastic nonlinear systems arouse broad interests. According to It\^o lemma, stochastic different equations are adopted to present stochastic processes, such as,
\begin{equation} \label{eq:paper13.stochasticNonlinearSysm}
d x = f(x,u)dt+h_w(x,u)dw(t),
\end{equation}
where $w\in\mathbb{R}^r$ is r-dimensional independent standard Wiener process,
$f(x,u):\mathbb{R}^n\times\mathbb{R}^m \to \mathbb{R}^n$ and $h_w(x,u):\mathbb{R}^n\times\mathbb{R}^m \to \mathbb{R}^{n\times r}$ are locally Lipschitz for all $t\geq 0$ satisfying $f(0,0)=0$ and $h_w(0,0)=0$. 

Different from the disturbance $d_i$ in system \eqref{eq:paper13.backstepping.boundedDisturbance.chen2011_adaptive.system} whose amplitude is assumed to be bounded, the magnitude of the disturbance in a stochastic system can be arbitrarily large in sufficiently long period. Hence, stabilities and properties are defined in probability. Unlike deterministic systems, it is impossible to prove the boundness. The equilibrium point $x(0)=0$ of (\ref{eq:paper13.stochasticNonlinearSysm}) is said to be 
\begin{enumerate}
\item Stable in probability if, for every $\varepsilon>0$ and $\delta>0$, there exists an $r$ s.t. if $t>t_0$, $|x_0|<r$ and $i_0\in S$, then $P\{|x(t)|>\varepsilon\}<\delta$; 
\item Asymptotically stable in probability if it is stable in probability and, for each $\varepsilon>0$, $x\in\mathbb{R}^{n}$ and $i_0 \in S$, there is $\lim\limits_{t\to\infty}P\{|x(t)|>\varepsilon\}=0$;
\item Bounded in probability if the random variable $|x(t)|$ are bounded in probability uniformly in $t$, i.e., $\lim\limits_{t\to\infty} \sup\limits_{t>t_0} P\{|x(t)|>R\}=0$ \cite{wu:2009backstepping}. 
\end{enumerate}

The error state is defined as $dz_1 = dx_1 - d x_{1d} =  dx_1 -  \dot{x}_{1d} dt$.
To prove the stability of a stochastic system, an infinitesimal generator for a stochastic system is adopted, given by
\begin{equation}
\mathcal{L}V(x) = \frac{\partial V}{\partial t} + \frac{\partial V}{\partial x} f + \frac{1}{2}\text{Tr}[h_w^\top\frac{\partial^2 V}{\partial x^2} h_w],
\end{equation}
where the higher-order Hessian term $\frac{1}{2}\text{Tr}[h_w^\top\frac{\partial^2 V}{\partial x^2} h_w]$ is due to the stochastic noises and $\text{Tr}(\cdot)$ is the trace operator. To adjust the newly-involved term $\frac{\partial^2 V}{\partial x^2}$, quartic Lyapunov functions are always chosen, e.g., $\frac{1}{4}z_i^4$ and $\frac{1}{4} \log \left(\frac{k_{bi}^4}{k_{bi}^4-z_i^4}\right)$. While the functions accounting for the errors of parameter estimates are still quadratic, e.g., $\frac{1}{2} \tilde{\theta}_i^\top \tilde{\theta}_i$. Consequently, $z_1^2$ appears in $\mathcal{L}V$. Applying Young's inequality generates a term with $z_1^4$ which can be compensated by the virtual control law $\alpha_i$. 

Suppose a positive definite, radically unbounded, twice continuously differentiable function $V$ exists s.t. $\mathcal{L}V(x)$ is negative definite, then the equilibrium $x=0$ is globally asymptotically stable in probability \cite{khasminskii:2011stochastic}. Additionally, the solution of the system, similar to Lemma~\ref{lemma:paper13.boundedV}, is bounded in probability suppose there exist two constants $\gamma>0$ and $\delta>0$ such that $\mathcal{L}V(x) \leq -\gamma V(x,t) + \delta$, for all $x\in\mathbb{R}^n$ and $t>t_0$ and $E[V(x,t)] \leq V(x_0) e^{-\gamma t} + \frac{\delta}{\gamma}$ \cite{psillakis:2007nn}. 
If Nussbaum-type function is adopted, $\mathcal{L}V(x) \leq -\gamma V(x,t) + \frac{1}{c} (g_n \mu_1 \mathcal{N}(\chi) + 1)\dot{\chi} \delta$ results in $E[V(x,t)] \leq V(x_0) e^{-\gamma t} + \frac{\delta}{\gamma} +  \frac{\sigma}{c}$ where $\sigma = \sup\int_{0}^{t}|E(g_n \mu_1 \mathcal{N}(\chi) + 1)\dot{\chi} e^{\gamma \tau}|d\tau$ \cite{wang:2016adaptiveHysteresis}. The finite-time stability is guaranteed by $\mathcal{L}V+c V^\alpha(x) \leq 0$ \cite{yin:2011finite}.

\section{Discussion and conclusion}
Backstepping design approaches have been significantly developed since the first appearance. Numerous modifications and innovations are proposed to achieve recursive Lyapunov-based controller design to increasingly complex and uncertain systems. This survey reviews the state-of-the-art development of several widely applied elegant methods and adaptive control using backstepping-like approaches to various nonlinearities are presented. The complex nonlinear systems are transferred into the integration of a number of simple problems, and each sub-problem is solved by the corresponding assumptions and approaches. To sum up, backstepping, together with its extensions, is a systematic and modularized control design approach.


Cancellation is the keyword in the backstepping-like methodology. The nonlinearities result in weak convergence, undesired inaccuracy and oscillation, and even instability. To cancel the uncertain and unknown parts, each method relays on some specific assumptions on boundness and parametric separation. Increasingly restrictive assumptions are needed for higher nonlinearities. The Lyapunov function for each step can be a sum of several Lyapunov functions for specific complexities, including tracking errors, approximation errors, smoothness errors, etc. When the appropriate assumptions and Lyapunov functions are invoked, the virtual control laws, adaptive update laws, and final control input can be designed to compensate the state-relevant terms in the time derivative of the LFC. 

In practical applications, there are many systems other than the strict-feedback form. When the simplifications are not satisfied, nonlinearities and complexities appear. 
Since the system uncertainties and complexities are not perfectly known, the Lyapunov stability criteria in a form of Lemma~\ref{lemma:paper13.boundedV} play important roles in the recursive design procedures and stability proofs. 
Noncancelable positive parts are moved to $\delta$ in \eqref{eq:paper13.boundedVLemma.V} and $\gamma$ is chosen as the minimum value of groups of control and update gains. 
As to exactly cancel all unknown terms is impossible, boundness is always guaranteed rather than the asymptotic or exponential stabilities at the equilibrium point.
With growing complexities and uncertainties, the magnitude of $\delta$ in \eqref{eq:paper13.boundedVLemma.V} increases with a sum of positive components. The control gains have to be sufficiently large to ensure a satisfying tracking performance. However, the real value of $\delta$ is probably much less than $\delta$ we received, resulting in a conservative design and a waste of actuator capability.

Most developments belong to two foremost approaches, i.e., robustness-based approach and estimation-based approach. 
\begin{itemize}
	\item The main idea of the robustness-based method is to largely compensate all nonlinearities by some parameter-separation assumptions, combing with $\dot{V}=-\gamma V+\delta$ in Lemma~\ref{lemma:paper13.boundedV}. The positive terms, which cannot be canceled but bounded, are accumulated in the final $\delta$. Then, the effects of system uncertainties can be reduced with high control gains and adaptive coefficients. However, such aggressive control would require expensive hardware, and cause instability. 
	Hence, the value of $\gamma$ must be enhanced to achieve similar performance. Consequently, the control and adjust laws are growingly aggressive with the increasing system dimension $n$. To enhance the adaption speed, large gains in the adaptive update laws may result in instabilities.
	\item The approximation-based approach employs estimators based on the existed measurements. All nonlinearities are assumed to be approximated. Increasing the order of the estimator is the solution to increasingly complex nonlinearities. Suppose that the unknown parts are well estimated, the nonlinearities can be canceled. However, exosystem, NN, and FLS all fail to estimate the high-frequency uncertainties. Though the NN/FLS-based approaches have high flexibility to approximate all sorts of bounded smooth functions, the learning process is time-consuming. 
\end{itemize}

The primary drawback of the backstepping approach is the complexity of the controller design and stability analysis, which is caused by increasing the system dimension and repeated differentiation of the virtual control laws. 
Overparameterization is still a very important problem. The case studies in most works are concentrated on the system with a degree equal and less than three. Estimation-based approaches, such as NN, and adaptive control laws introduce significant amounts of states to the entire control system, consequently, influencing the computation speed, learning time, and reaction speed. Attempts have been made to overcome the overparameterization problem during backstepping design procedures, e.g., tuning function, dynamic surface control, and command filtered backstepping. 

The unknown terms are assumed to be either bounded or a product of unknown constant vectors and known functions. Systematic methods to find $\phi_i (\bar{x}_i)$ is still challenging. Though the appearance of $\phi_i (\bar{x}_i)$ can be received with careful tuning, it is difficult to ensure the proposed functions fit all the uncertainties. Additionally, the design procedures would be cumbersome and the computational time boosts if too many unnecessary components invoked. Approaches with lower gains in control laws and adjusting laws will be valuable.

Another phenomenon of studies on backstepping is the combination-based innovations. 
Recalling the memory of cocktail and mixology, the recipe gets a new name with various mixtures of ingredients. However, they are still beverages without distinction in the cooking methods. This is similar to today's publications on backstepping design. Some elegant methods are developed. Notwithstanding, most studies are just never-appeared-before combinations of several complexities using an integration of existing approaches that handle the corresponding nonlinearities. On the other hand, this is a preferred feature by the control designers, which makes the design modularized.

Backstepping has shown a great promise in numerous applications since initial development. Novel methods have extended its applications to increasingly complex systems. We believe that it will be applied to new future challenges by the involvement of theoretical innovations in relevant realms.

\bibliographystyle{elsarticle-num}
\bibliography{TK8111}


\end{document}